\documentclass[11pt,reqno]{amsart}
\usepackage[utf8]{inputenc}
 
\usepackage{amsmath,amsthm,amssymb}

\usepackage{graphicx}
\usepackage{caption, subcaption}
\usepackage{tikz}
\usepackage{enumitem}
\graphicspath{ {images/} }

\usepackage{geometry}
\usepackage{color}
\usepackage{cleveref}
\usepackage[sort&compress,numbers]{natbib}
\usepackage{todonotes}
\usepackage{comment}

\newtheorem{theorem}{Theorem}

\newtheorem{prop}[theorem]{Proposition}
\newtheorem{lemma}[theorem]{Lemma}

\newtheorem{remark}[theorem]{Remark}
\newtheorem{defn}[theorem]{Definition}

\numberwithin{equation}{section}
\numberwithin{theorem}{section}

\newcommand{\coup}{{g}}
\newcommand{\iu}{\mathrm{i}\mkern1mu} 
\renewcommand{\d}{\mathrm{d}}
\newcommand{\p}{\partial}

\newcommand{\bq}{{\bf q}}

\newcommand{\bk}{{\bf k}}

\newcommand{\bx}{{\bf x}}
\newcommand{\by}{{\bf y}}

\newcommand{\eps}{\epsilon}

\renewcommand{\Re}{\operatorname{Re}}
\renewcommand{\Im}{\operatorname{Im}}
\newcommand{\tr}{\operatorname{tr}}

\newcommand{\R}{\mathbb{R}}

\newcommand{\C}{\mathbb{C}}

\newcommand{\A}{\mathcal{A}}
\newcommand{\B}{\mathcal{B}}

\newcommand{\G}{\mathcal{G}}
\renewcommand{\L}{\mathcal{L}}

\newcommand{\ds}{\displaystyle}

\DeclareMathOperator\supp{supp}

\newcommand{\Glog}[2]{B_{#1}^{#2}}
\newcommand{\Alog}[2]{\mathcal{B}_{#1}^{#2}}

\title[Nonlocal PDEs and Quantum Optics: Bound States and Resonances]{Nonlocal Partial Differential Equations and Quantum Optics: Bound States and Resonances}

\author[Hiltunen]{Erik Orvehed Hiltunen}
\address{Department of Mathematics, Yale University, New Haven, CT, USA }
\email{erik.hiltunen@yale.edu}

\author[Kraisler]{Joseph Kraisler}
\address{Department of Applied Physics and Applied Mathematics, Columbia University, New York, NY, USA}
\email{jek2199@columbia.edu}

\author[Schotland]{John C. Schotland}
\address{Department of Mathematics and Department of Physics, Yale University, New Haven, CT, USA}
\email{john.schotland@yale.edu}

\author[Weinstein]{Michael I. Weinstein}
\address{Department of Applied Physics and Applied Mathematics and Department of Mathematics, Columbia University, New York, NY, USA}
\email{miw2103@columbia.edu}
\date{\today}

\begin{document}

\begin{abstract}
We consider the quantum optics of a single photon interacting with a system of two level atoms.
This leads to the study of a nonlinear eigenproblem for a system of nonlocal partial differential equations.
We investigate two classes of solutions: bound states are solutions which decay at infinity 
while resonance states have locally finite energy and satisfy a non-self-adjoint outgoing radiation condition at infinity.  We have found necessary and sufficient conditions for the existence of bound states, along with an upper bound on the number of such states.  We have also considered these problems for atomic models with small high contrast inclusions. In this setting, we have derived asymptotic formulas for the resonances.
 Our results are illustrated with numerical computations. 
\end{abstract}

\maketitle

\section{Introduction}

This paper is the first in a series that is concerned with a class of mathematical problems that arise in quantum optics. We consider the following model for the interaction between a quantized field and a collection of two-level atoms, as introduced in~\cite{kraisler_2022}.
The Hamiltonian $\mathcal H$ of the system is given by
\begin{align}
\mathcal H=\mathcal H_F + \mathcal H_A + \mathcal H_I , 
\end{align}
where $\mathcal H_F$ is the Hamiltonian of the field, $\mathcal H_A$ is the Hamiltonian of the atoms and $\mathcal H_I$ is the interaction Hamiltonian.
The Hamiltonian of the field is of the form 
\begin{align}
\label{H_F}
\mathcal  H_F = \hbar c\int_{\R^d}  (-\Delta)^{1/2}\phi^{\dagger}(\bx)\phi(\bx) \d\bx,
\end{align}
where $\phi(\bx)$ is a Bose scalar field that obeys the commutation relations
\begin{align}
\label{commutation}
[\phi(\bx),\phi^{\dagger}(\bx')]&=\delta(\bx-\bx') , \\
[\phi(\bx),\phi(\bx')]&=0 .
\end{align}
The operator $(-\Delta)^{1/2}$ is nonlocal and has the Fourier representation
\begin{align}
    (-\Delta)^{1/2}f = \int_{\R^d} e^{\iu \bk\cdot\bx}\vert\bk\vert \hat f(\bk)\frac{\d \bk}{(2\pi)^d}.
\end{align}
Here, the Fourier transform $\hat f(\bk)$ of a function $f(\bx)\in L^2(\R^d)\cap L^1(\R^d)$ is defined as
\begin{align}
    \hat f(\bk) = \int_{\R^d} e^{-\iu\bk\cdot\bx} f(\bx) \d \bx.
\end{align}
The Hamiltonian of the atoms is given by
\begin{align}
\mathcal H_A = \hbar \Omega\int_{\R^d} \rho(\bx)\sigma^{\dagger}(\bx)\sigma(\bx) \d\bx ,
 \label{HA}
\end{align}
where $\Omega $ is the atomic resonance frequency, $\rho(\bx)$ is the number density of the atoms. The Fermi field $\sigma(\bx)$, defined on $\supp(\rho)$,  obeys the anticommutation relations
\begin{align}
\label{anticommutation}
\{\sigma(\bx),\sigma^{\dagger}(\bx')\}&=\frac{1}{\rho(\bx)}\delta(\bx-\bx') , \\
\{\sigma(\bx),\sigma(\bx')\}&=0 .
\end{align}
The Hamiltonian describing the interaction between the field and the atoms is taken to be
 \begin{align}
\mathcal H_I = \hbar \coup\int_{\R^d}  \rho(\bx)\left( \phi^\dagger(\bx)\sigma(\bx)+\sigma^{\dagger}(\bx)\phi(\bx)\right)\d\bx ,
\end{align}
where $\coup$ is the strength of the atom-field coupling.

We suppose that the system is in a single-excitation state of the form
\begin{align}
\label{one_photon_state}
\vert\Psi\rangle = \int_{\R^d} \d\bx \left[\psi(\bx,t)\phi^{\dagger}(\bx) + \rho(\bx)a(\bx,t)\sigma^{\dagger}(\bx)\right]\vert 0\rangle ,
\end{align}
where $\vert 0\rangle$ is the combined vacuum state of the field and the ground state of the atoms. Here $a(\bx,t)$ denotes the probability amplitude for exciting an atom at the point $\bx$ at time $t$, and $\psi(\bx,t)$ is the  amplitude for creating a photon. 
The dynamics of $\vert\Psi\rangle$ is governed by the Schrödinger equation
\begin{align} 
\label{eq:schr}
\iu \hbar\partial_t\vert\Psi\rangle = \mathcal H\vert\Psi\rangle .
\end{align}
It follows that $a$ and $\psi$ obey the system of \emph{nonlocal} equations 
\begin{align}
\label{eq:b4}
\iu\partial_t\psi & = c(-\Delta)^{1/2}\psi + \coup\rho(\bx) a , \\
\iu \rho(\bx)\partial_t a & = \coup\rho(\bx)\psi + \Omega\rho(\bx) a .
\label{eq:b5}
\end{align}
We note that $\vert\Psi\rangle$ is normalized so that the amplitudes obey the conservation law
\begin{align}\label{eq:conserve}
\int_{\R^d}  \left(\vert\psi(\bx,t)\vert^2 + \rho(\bx)\vert a(\bx,t)\vert^2\right)\d\bx = 1
\ .
\end{align}
A well-posedness result for
solutions of the initial value problem for \eqref{eq:b4}-\eqref{eq:b5} is presented in Appendix \ref{app:EandU}.

The system (\ref{eq:b4})--(\ref{eq:b5}) describes the dynamics of a single-excitation state. It has been studied in a number of settings~\cite{kraisler_2022}. These include spontaneous emission from a single atom with density $\rho(\bx)=\delta(\bx)$ and a system of atoms with constant density. The long-time asymptotic behavior for a finite number of atoms has also been investigated, and a crossover from exponential decay to algebraic decay at long times has been predicted~\cite{hoskins_2021}. The prediction has been confirmed by a fast, high-order numerical method for solving (\ref{eq:b4})--(\ref{eq:b5})~\cite{hoskins_2023}. Finally, the case of a random medium has been explored. Here the dynamics of the field and atomic probability amplitudes have been studied by means of the asymptotics of the Wigner transform. At long times and large distances,  the corresponding average probability densities can be determined from the solutions to a pair of kinetic equations~\cite{kraisler_2022}.

We now consider the time-harmonic solutions to the time-dependent equations \eqref{eq:b4}--\eqref{eq:b5} of the form
\begin{align}\label{eq:t-harmonic}
    \begin{pmatrix}
    a(\bx,t) \\ \psi(\bx,t) 
    \end{pmatrix} = e^{-\iu\omega t}\begin{pmatrix}
    a(\bx) \\\psi(\bx)
    \end{pmatrix}.
\end{align}
We thus obtain the time-independent equations
\begin{align} 
	c(-\Delta)^{1/2}\psi + g\rho(\bx)a &= \omega \psi, \label{eq:system1} \\
	g\rho(\bx) \psi+ \Omega \rho(\bx) a &= \omega \rho(\bx) a . \label{eq:system2}
\end{align}
{For $\omega\ne\Omega$ we may, on the support of $\rho$, use \eqref{eq:system2} to eliminate $a$ from \eqref{eq:system1}, to obtain a  equation for}  $\psi$:
\begin{equation}\label{eq:BS}
	\left((-\Delta)^{1/2} - \frac{\omega}{c}\right) \psi = - \frac{g^2}{c(\omega-\Omega)}\rho(\bx)\psi.
\end{equation}
Then, $a$ can be recovered from the formula
\begin{align}\label{eq:a}
    a(\bx) = \frac{g}{\omega-\Omega}\psi(\bx),\ \ \bx\in\supp(\rho).
\end{align}

It will prove convenient to rewrite the system \eqref{eq:system1}--\eqref{eq:system2} in a more symmetric form.
To proceed, we make the change of variables 
\begin{align}
\phi(\bx) = a(\bx)\sqrt{\rho(\bx)} .   
\end{align}
Eqs.~\eqref{eq:b4}--\eqref{eq:b5} can then be written in the form of
the eigenproblem 
\begin{align}\label{eigenproblem}
    H\begin{pmatrix}
    \psi\\\phi
    \end{pmatrix} = \omega\begin{pmatrix}
    \psi\\\phi
    \end{pmatrix},
\end{align}
where the Hamiltonian $H$ is given by
\begin{align}\label{eq:H}
H=\begin{pmatrix}c(-\Delta)^{1/2} & g\sqrt{\rho(\bx)}\\ g\sqrt{\rho(\bx)} &  \Omega \end{pmatrix}.
\end{align}

In this paper, we study two eigenvalue problems {associated with  \eqref{eigenproblem}}: the bound state and resonance eigenvalue problems.
\begin{enumerate}
    \item \textit{The bound state problem}.  A {\it bound state} is a non-zero pair $(\psi,\phi)\in L^2(\R^d)$ and frequency $\omega$ which satisfies the system \eqref{eq:system1}--\eqref{eq:system2}.  By self-adjointness of $H$, 
bound states of \eqref{eq:system1}--\eqref{eq:system2} can occur only for $\omega$ real. Below we shall see, for a physically relevant class of  densities $\rho$, that additionally we must have $\omega<0$. Note that if $\psi\in L^2(\R^d)$ and $\omega\ne\Omega$ are such that  \eqref{eq:BS} holds, then $(\psi,\phi)$, where 
$\phi=g(\omega-\Omega)^{-1}\sqrt{\rho}\psi$,  is a bound state with frequency $\omega$. 
    \item \textit{The resonance problem}. A {\it resonance state}  is a non-zero pair $(\psi,\phi)$ and frequency $\omega\in\mathbb{C}$ which satisfies the system \eqref{eq:system1}-\eqref{eq:system2}, such that $\psi$ is {outgoing} in the sense that $\psi$ can be represented through the outgoing resolvent of $(-\Delta)^{1/2}$ (see \Cref{sec:int} and \Cref{sec:rad}). For real and positive $\omega$, the outgoing condition is equivalent to the radiation condition
    \begin{equation}\lim_{|\bx|\to \infty} |\bx|^{\frac{d-1}{2}}\left(\frac{\p \psi}{\p |\bx|} - \frac{\iu \omega}{c} \psi\right) = 0. \label{eq:smfld}\end{equation} 
    We show below in \Cref{prop:im} that for resonance states, the condition $\Im(\omega)\leq 0$ holds. Whenever $\Im(\omega)< 0$, self-adjointness implies that resonance states are in $L^2_{\rm loc}$ but not in $L^2(\R^d)$.{Furthermore, by \eqref{eq:t-harmonic} resonant states are exponentially decaying as $t\to+\infty$. Resonances give information about the rate of escape of energy, as time advances, from a fixed compact set.} There are no resonance states for $\Re(\omega) <0$; see \Cref{rmk:neg} below.
\end{enumerate}

Our main results may be summarized as follows. 
Theorem~\ref{thm:boundstates} guarantees the existence of bound states for compactly supported bounded densities $\rho\in L^{\infty}_{\rm comp}(\mathbb R^d)$. 
For dimension $d=1$ there is at least one bound state. In contrast, for dimensions $d\ge2$,  there exist densities for which there are no bound states. There is also a sufficient condition for $d\ge2$ that guarantees the existence of at least one bound state. {To prove this result we reformulate the spectral problem in terms of a family of compact integral operators, and apply a Birman-Schwinger principle analogous to that arising in the setting of nonrelativistic quantum mechanics \cite{ReedSimonIV}.} 
Theorem~\ref{thm:NBS} provides an upper bound on the number of bound states for compactly supported smooth densities, which by an approximation argument extends to densities $\rho\in L^d(\mathbb R^d)$. The proof makes use of a generalized Feynman-Kac formula
\cite{Daubechies1983}.
Next, we consider the resonances and bound states of small, high-contrast inclusions, densities of the form $\rho(\bx)= \rho_0\chi_D(\bx)$, where $D$ is a bounded domain. In addition, $\rho_0$ is assumed to scale as $1/\epsilon$ as $\epsilon\to 0$, where the diameter of $D$ is proportional to $\epsilon$. The analysis makes use of an integral representation for the amplitude $\psi$ in terms of a suitable Green's function. This representation leads to a {{\it nonlinear eigenvalue problem} (an eigenvalue problem for a linear operator with nonlinear dependence on the spectral parameter). In Theorems~\ref{thm:3D} and \ref{thm:2D}, we calculate its eigenvalues (bound state and resonance frequencies)  asymptotically in the limit $\epsilon\to 0$.} These results rely on the hypothesis that the eigenvalues are simple. At least for the lowest eigenvalue of the limiting problem, this hypothesis can be verified.
Numerical computation of the eigenvalues is also performed. The corresponding problem in one dimension can also be studied, but requires a different scaling of the density $\rho_0$ with $\epsilon$. Theorem~\ref{thm:1Dlog} describes the asymptotics of the eigenvalues in this setting.

The paper is organized as follows. Section~\ref{sec:bound} considers the problem of bound states. Section~\ref{sec:incl} discusses the problem of {bound states and resonances for} high-contrast inclusions. {We summarize our conclusions and discuss further directions to investigate in Section~\ref{sec:conclusions}.} The appendices present various properties and computations of Green's functions, an outline of perturbation theory for holomorphic operator-valued functions, and recall some results from functional analysis that are used throughout the paper.

\section{Bound states} \label{sec:bound}

Suppose $\omega\ne\Omega$. Then if $\psi\in L^2(\R^d)$ is a solution of \eqref{eq:H} and we define $\phi$ by \eqref{eq:a}, then $(\psi,\phi)^\top$ is an eigenstate for the eigenvalue problem \eqref{eigenproblem}. Conversely, any eigenstate $(\psi,\phi)^\top$ 
 of \eqref{eigenproblem} gives rise to an $L^2(\R^d)$ solution of \eqref{eq:H}.
 
In this section we determine sufficient conditions for the existence of eigenstates of \eqref{eigenproblem}.
We assume that $\rho(\bx)$ satisfies the following hypotheses:
\begin{enumerate}[label=\textbf{H.\arabic*}]
    \item $\rho(\bx)\geq 0$ (non-negative). \label{Hyp1}
    \item $\supp{\rho}$ is compact (compact support). \label{Hyp2}
    \item $\rho\in L^{\infty}(\R^d)$ (essentially bounded). \label{Hyp3}
\end{enumerate}
Note that since $\rho$ vanishes outside a compact set, the essential spectrum of $H$ belongs to $[0,\infty)$. We do not expect there to be embedded eigenvalues. Indeed, under smoothness assumptions on $\rho$ it has been shown \cite{Lorinczi2022} that the operator 
\begin{align*}
        H(\omega) = c(-\Delta)^{1/2} + \frac{g^2\rho}{\omega-\Omega} .
    \end{align*}
    has no positive spectrum. Since any bound state arises from 
    an $L^2(\R^d)$ solution of $H(\omega)\psi = (\omega/c)\psi$, it follows, under such hypotheses on $\rho$, that there are no positive eigenvalues $\omega$.
Hence, we focus on conditions for the  existence of eigenstates with energies $\omega < 0$.


\subsection{Necessary condition for bound states}
We begin by stating a necessary condition for bound states. This condition relates any eigenvalue $\omega<0$ to the magnitude of the atom--field coupling strength $g^2\|\rho\|_{\infty}$. If the latter is small, then any negative eigenvalue must necessarily be close to $0$.
\begin{prop}
	Let $H$ be defined as in \eqref{eq:H} and let $\rho$ be a bounded, compactly supported, positive density. Assume that $H$ has a negative eigenvalue $\omega < 0$ with corresponding eigenstate $\psi\in L^2(\R^d)$. Then 
	\begin{equation}
		g^2\|\rho\|_{\infty} \geq \omega(\omega-\Omega).
	\end{equation}
\end{prop}
\begin{proof}
	For $\omega <0$, the operator $\left((-\Delta)^{1/2} - \frac{\omega}{c}\right)$ is invertible. Thus we have a Lippmann-Schwinger representation of \eqref{eq:BS} of the form
	\begin{equation}\psi(\bx) = -\frac{g^2}{c(\omega-\Omega)}\int_{\R^d} G^{\omega/c}(\bx-\by) \rho(\by) \psi(\by) \d \by.\end{equation}
	Here $G^{\omega/c}$ is the Green's function for negative $\omega$, given in \Cref{app:G}. From \Cref{lemma:Greensdecay}, it is clear that $G^{\omega/c}$ is real, positive and in $L^1(\R^d)$. Additionally, $\psi\in L^1$ since the above formula expresses it as a convolution of two $L^1$ functions. If we integrate over $\R^d$ and estimate the integral, we therefore have
	\begin{equation}\label{eq:est}\int_{\R^d} |\psi| \d \bx \leq -\frac{g^2}{c(\omega-\Omega)}\left(\int_{\R^d}G^{\omega/c}(\bx)\d \bx\right)\left(\int_{\R^d}\rho|\psi| \d \bx\right).\end{equation}
	Here we note that 
	\begin{equation}\int_{\R^d} G^{\omega/c}(\bx) \d \bx = \hat{G}^{\omega/c}(0),\end{equation}
	where $\hat{G}^{\omega/c}(\bk)$ denotes the Fourier transform of $G^{\omega/c}$, given by
	\begin{equation}\hat{G}^{\omega/c}(\bk) = \frac{1}{|\bk|-\frac{\omega}{c}}.\end{equation}
	Therefore $$\displaystyle\int_{\R^d} G^{\omega/c}(\bx) \d \bx =  -\frac{c}{\omega}.$$
	It follows from \eqref{eq:est} that 
	\begin{equation}\int_{\R^d} |\psi| \d \bx \leq \frac{g^2\|\rho\|_{\infty}}{\omega(\omega-\Omega)}\int_{\R^d}|\psi| \d \bx,\end{equation}
 and since $\psi\ne0$, we find that $g^2\|\rho\|_{\infty} \geq \omega(\omega-\Omega)$. Observe, in particular, that this condition is independent of the dimensionality $d$.\end{proof}

\subsection{Sufficient conditions for bound states} \label{sec:bound_suff}
The next theorem is the main result of this section. It guarantees the existence of bound states in dimension 1. Furthermore, it displays a family of densities illustrating that for sufficiently small densities there are no bound states in dimension 2 and higher, and that one can expect for $d\ge2$ that if the density is sufficiently large, there exists at least one bound state.
\begin{theorem} \label{thm:boundstates}
Consider the self-adjoint operator $H$, given by Eq.~(\ref{eq:H}), with dense domain $H^1(\R^d)\times L^2(\R^d)$. We have the following:
\begin{enumerate}
    \item If $d=1$, and the density $\rho$ satisfies hypotheses \ref{Hyp1}-\ref{Hyp3}. Then $H$ has at least one bound state with negative frequency $\omega < 0$.
    \item For $d\geq 2$, consider the family of piecewise constant densities $\rho= \rho_0\chi_{[-R,R]^d}$. If \begin{align}
            \frac{2g^2\rho_0 R}{\Omega} < S_d,
    \end{align}
    where $S_d$ is given in Eq.~(\ref{eq:SN}), then there are no negative-energy eigenstates.
    
    \item If $d\geq 2$, $\omega < 0$, and $\rho = \rho_0\chi_{[-R,R]^d}$, then there exists a constant $K$ independent of $g,\rho_0,\Omega$, and $R$ such that if
    \begin{align}
        g^2\rho_0 \geq K(\Omega+\vert\omega\vert)\left(\frac{\pi \tilde c}{2R}+\vert\omega\vert\right),
    \end{align}
    then $H$ has at least one bound state in the interval $(-\infty,\omega]$. In particular, if 
    \begin{align}\label{eq:Omegabound}
    \frac{2g^2\rho_0R}{\Omega c} > K\pi,
    \end{align}
    there is at least one bound state.
\end{enumerate}
\end{theorem}
\begin{remark}
The hypothesis that $\rho(x)$ be compactly supported can be relaxed in part $1$ of the theorem above. All we require is that the operator $K_{\omega}[\rho]$ defined in Eq.~(\ref{eq:BSop}) is compact. For example, a positive, but sufficiently rapidly decaying $\rho(x)$ will also satisfy this condition.
\end{remark}
In order to prove the theorem, we will rewrite the eigenvalue problem for $H$ given in Eq.~(\ref{eq:H}) as an equivalent problem involving compact integral operators. This is analogous to the Birman-Schwinger principle used to study the Schr\"odinger equation ~\cite{Birman61,Schwinger61}. First we define a family of operators $K_\omega[\rho]:L^2(\R^d)\to L^2(\R^d)$, indexed by $\omega < 0$, and depending on the density $\rho$: 
\begin{align} \label{eq:BSop}
     K_\omega[\rho]= \frac{g^2 \rho^{1/2}\left(c(-\Delta)^{1/2}+\vert\omega\vert\right)^{-1} \rho^{1/2}}{\Omega+\vert \omega\vert}.
\end{align}
Note that $K_\omega[\rho]$ is self-adjoint and compact on $L^2(\R^d)$.
\begin{theorem} \label{BirmanSchwinger2}
Let $\omega<0$. The number of eigenvalues of $H$ in the interval $(-\infty,\omega]$ is equal to the number of eigenvalues of $ K_{\omega}[\rho]$ in the interval $[1,\infty)$
\end{theorem}

\begin{proof}[Proof of Theorem \ref{BirmanSchwinger2} ]
Let $\omega<0$. We claim that $1$ is an eigenvalue of $ K_{\omega}[\rho]$ if and only if $\omega$ is an eigenvalue of $H$. We sketch the proof; see ~\cite{Birman61,Schwinger61}.  Let $H\phi = \omega\phi$, then the first component $\phi_1$ satisfies
\begin{align}
    c(-\Delta)^{1/2}\phi_1 + \frac{g^2\rho}{\omega-\Omega}\phi_1 = \omega\phi_1.
\end{align}
or
\begin{align}
    \left(c(-\Delta)^{1/2}+\vert\omega\vert\right)\phi_1 = \frac{g^2\rho}{\Omega+\vert \omega\vert}\phi_1.
\end{align}
The change of variables
\begin{align}
    \psi = \rho^{1/2}\phi_1,
\end{align}
and inverting  $\left(c(-\Delta)^{1/2}-\omega\right)$ yields a symmetrized form of the eigenvalue problem:
\begin{align}
    \psi = \frac{g^2}{\Omega+\vert \omega\vert}\rho^{1/2}\left(c(-\Delta)^{1/2}+\vert\omega\vert\right)^{-1}\rho^{1/2}\psi = K_{\omega}[\rho]\psi.
\end{align}
Since  $ K_{\omega}[\rho]$ is self-adjoint and compact, its eigenvalues $(\mu_n(\omega))_{n\ge1}$ can be found via the max-min principle; see \cite{ReedSimonIV} Theorem XIII.1.  $\omega$ is an eigenvalue of $H$ exactly when $\mu_n(\omega)=1$ for some $n\ge1$. Since $\|K_{\omega}[\rho]\|\to 0$ as $\omega\to -\infty$, it follows for each $n\ge1$, that $ \mu_{n}(\omega)\to 0$. So suppose there are $k$ eigenvalues of $ K_{\omega}[\rho]$ in the interval $[1,\infty)$. Then as $\omega\to -\infty$ each of these must cross the value $1$ at some value less than or equal to $\omega$, say $\omega_1,\cdots,\omega_k$. This gives rise to $k$ eigenvalues of $H$ in $(-\infty,\omega]$. Conversely, as $\mu_n(\omega)$ is monotonically decreasing in $\omega$, one observes that these are the only such eigenvalues. 
\end{proof}

\begin{proof}[Proof of Theorem \ref{thm:boundstates}]
By the max-min principle $K_\omega[\rho]$ has an eigenvalue in $[1,\infty)$ if we can produce a trial function $\phi$, with $\|\phi\|=1$ such that $\langle \phi,  K_\omega[\rho_1]\phi\rangle>1$.
We start with a general estimate valid in any dimension $d\geq 1$. Let $\rho_1(\bx)$ be a square density of width $2R > 0$ and height $\rho_0> 0$ centered at the origin
\begin{align}\label{eq:squaredensity}
    \rho_1(\bx) = \rho_0\chi_{[-R,R]^d}.
\end{align}
Let $K_{\omega}[\rho_1]$ be the operator defined in Eq.~(\ref{eq:BSop}). For any $\phi\in L^2(\R^d)$
\begin{align}
\begin{split}
    \langle \phi,  K_\omega[\rho_1]\phi\rangle &= \frac{g^2}{\Omega+\vert\omega\vert}\langle \rho_1^{1/2}\phi,(c(-\Delta)^{1/2}+\vert\omega\vert)^{-1}(\rho_1^{1/2}\phi)\rangle\\ 
     &= \frac{g^2\rho_0}{\Omega+\vert\omega\vert}\langle \phi\chi_{[-R,R]^d},(c(-\Delta)^{1/2}+\vert\omega\vert)^{-1}(\phi\chi_{[-R,R]^d})\rangle\\ 
    &=\frac{g^2\rho_0}{\Omega+\vert \omega\vert}\int_{\R^d}\vert \widehat{(\phi\chi_{[-R,R]^d})}(\bq)\vert^2\frac{1}{c\vert \bq\vert +\vert\omega\vert}\frac{\d\bq}{(2\pi)^d}. \label{eq:BS3}
\end{split}
\end{align}
We choose $\phi$ so that $\widehat{(\phi\chi_{[-R,R]^d})}(\bq)$ is continuous and nonzero at $\bq =0$. For example, we may pick $\phi\in C_0^{\infty}(\R^d)$ and such that
\begin{align*}
    \supp\phi \subset [-R,R]^d, \quad   \widehat{\phi}(0)= \int_{\R^d}\phi\ \d\bx \neq 0,\quad \int_{\R^d}\vert\phi\vert^2\d \bx = 1 .
\end{align*}
Then for $\delta >0$ and small, we obtain the lower bound,
\begin{align}
\begin{split}
       \langle \phi,  K_\omega[\rho_1]\phi\rangle &\geq  \frac{g^2\rho_0}{\Omega+\vert \omega\vert}\int_{|\bq|\leq\delta}\vert \widehat{\phi}(\bq)\vert^2\frac{1}{c\vert \bq\vert +\vert\omega\vert}\frac{\d\bq}{(2\pi)^d} \\
        & = \frac{g^2\rho_0}{\Omega+\vert \omega\vert}\vert \widehat{\phi}(0)\vert^2\int_{\vert \bq\vert < \delta}\frac{1}{c\vert \bq\vert+\vert\omega\vert}\frac{\d\bq}{(2\pi)^d} \\
        & -\frac{g^2\rho_0}{\Omega+\vert \omega\vert}\int_{|\bq|\leq\delta}\frac{\vert \widehat{\phi}(\bq)\vert^2-\vert \widehat{\phi}(0)\vert^2}{c\vert \bq\vert +\vert\omega\vert}\frac{\d\bq}{(2\pi)^d} \\
       &\geq \frac{g^2\rho_0}{\Omega+\vert \omega\vert}\vert \widehat{\phi}(0)\vert^2\int_{\vert \bq\vert < \delta}\frac{1}{c\vert \bq\vert+\vert\omega\vert}\frac{\d\bq}{(2\pi)^d} - C\delta^d \label{eq:BSestimate} 
\end{split}
\end{align}
where the constant $C$  is independent of $\omega$. 

\noindent\textbf{Proof of Part 1 of Theorem \ref{thm:boundstates}: }\\
Let $d=1$ and $\rho(x)$ be essentially bounded and positive with compact support. In 1 dimension, $\rho_1(x)$ is given by
\begin{align}
    \rho_1(x) = \rho_0\chi_{[-R,R]}(x-x_0).
\end{align}
For an appropriate choice of the constants $\rho_0,R, x_0$ we have
\begin{align} \label{eq:monotonicity1}
    \rho_1(x) \leq \rho(x),
\end{align}
for every $x\in \R$. Then 
\begin{align} \label{eq:monotonicity2}
    K_{\omega}[\rho_1]\leq K_{\omega}[\rho],
\end{align}
as nonnegative, compact, self-adjoint operators on $L^2(\R^d)$. By (\ref{eq:BSestimate}), there exists $\phi\in L^2(\R^d)$, with $\|\phi\|_2=1$, such that 
\begin{align}
    \liminf_{\omega\nearrow 0} \langle \phi,  K_{\omega}[\rho_1]\phi\rangle = \infty,
\end{align}
For this choice of $\phi$, there exists $\omega_*<0$ for which 
\begin{align}
    \langle \phi,  K_{\omega_*}[\rho_1]\phi\rangle > 1.  
\end{align}
Hence $\mu_1(\rho_1,\omega_*)$, the largest eigenvalue of $K_{\omega_*},[\rho_1]$, satisfies
\begin{align}
    \mu_1(\rho_1,\omega_*) = \max_{\|f\|_2=1} \langle f,  K_{\omega_*}[\rho_1]f\rangle > 1.
\end{align}
Moreover, by (\ref{eq:monotonicity1}) we have that
\begin{align}
    \mu_1(\rho,\omega_*) > 1.
\end{align}
Theorem \ref{BirmanSchwinger2} now implies that   $ K_{\omega}[\rho]$ has at least one eigenvalue in the interval $[1,\infty)$. Therefore, (\ref{eq:BS}) has an $L^2(\R^d)$ solution with $\omega\in(-\infty,\omega_*]$.\\

\noindent\textbf{Proof of Part 2 of Theorem \ref{thm:boundstates}: }\\
Now suppose that we let $d\geq 2$ and let $\rho_1$ be given as in \eqref{eq:squaredensity}. We consider the quadratic form
associated to the operator $\tilde H(\omega)$,
\begin{align}
\begin{split}
    \langle \psi,\tilde H(\omega)\psi\rangle &= c\int_{\R^d}\psi(-\Delta)^{1/2}\psi + \frac{g^2\rho(x)}{\vert\omega\vert+\Omega}\vert\psi(x)\vert^2 \d x .
\end{split}
\end{align}

We next derive a lower bound on $ \langle \psi,\tilde H(\omega)\psi\rangle$. Using 
the Sobolev-type bound Theorem 8.4 of \cite{lieb2001analysis}
\begin{align}
    \bigl\langle f, (-\Delta)^{1/2}f\bigr\rangle \geq S_d\|f\|_q^2 ,\quad f\in H^{1/2}(\R^d),\quad d\ge2,
\end{align}
we have 
\begin{align}
\begin{split}
    \langle \psi,\tilde H(\omega)\psi\rangle 
    &\geq c\langle \psi,(-\Delta)^{1/2}\psi\rangle - \frac{2g^2\rho_0R}{\vert \vert\omega\vert+\Omega\vert}\left\| \vert\psi\vert^2 \right\|_{d/(d-1)}\\
    &\geq \left\| \vert\psi\vert^2 \right\|_{d/(d-1)}\left(cS_d-\frac{2g^2\rho_0R}{\Omega +\vert\omega\vert}\right)\\
    &\geq \left\| \vert\psi\vert^2 \right\|_{d/(d-1)}2g^2\rho_0  R\left(\frac{1}{\Omega}-\frac{1}{\Omega+\vert\omega\vert}\right).
\end{split}
\end{align}
Hence, if  $R$, $g$ and $\rho_0$ be chosen so that 
\begin{align}\label{eq:SN}
    \frac{2g^2\rho_0 R}{\Omega c} < S_d := \frac{d-1}{2}\bigl\vert\mathbb{S}^d\bigr\vert ^{1/d}, 
\end{align}
we have that 
\begin{align}
    \inf_{\|\psi\|=1}\langle\psi,\tilde H(\omega)\psi\rangle  \geq 0.
\end{align}
By the variational characterization of the ground state of $\tilde H(\omega)$, we obtain that any eigenvalue, $\lambda(\omega)$, of $\tilde H(\omega)$ must satisfy $\lambda(\omega)\ge0$. 
Since any eigenvalue $\omega$ of $H$ corresponds to a solution of $\lambda(\omega)=\omega$, we obtain a contradiction since we have assumed $\omega<0$.

\noindent\textbf{Proof of Part 3 of Theorem \ref{thm:boundstates}: }\\
Let $d\ge2$. We now show the existence of negative-energy bound states for sufficiently large densities. Once again let $\rho_1$ be as in \eqref{eq:squaredensity}.

We let $\phi(\bx)=(2R)^{-d/2}\chi_{[-R,R]^d}(\bx)$ in \eqref{eq:BS3} and derive a lower bound:
\begin{align}
\begin{split}
    \langle \phi,  K_{\omega}[\rho]\phi\rangle &= \frac{g^2\rho_0}{\Omega+\vert\omega\vert}\frac{1}{(2\pi)^d}\int_{\R^d}\vert\widehat{\phi \chi_{[-R,R]^d}}(\bq)\vert^2 \frac{1}{c\vert \bq\vert +\vert\omega\vert}\d\bq\\
    &=\frac{g^2\rho_0}{\Omega+\vert\omega\vert}\frac{1}{(2R)^d}\frac{1}{(2\pi)^d}\int_{\R^d}\vert\widehat{ \chi_{[-R,R]^d}}(\bq)\vert^2 \frac{1}{c\vert \bq\vert+\vert\omega\vert}\d\bq\\
     &=\frac{g^2\rho_0}{\Omega+\vert\omega\vert}\frac{1}{(2R)^d}\frac{1}{(2\pi)^d}\int_{\R^d}\prod_{i=1}^d\left(\frac{2\sin(q_i R)}{q_i}\right)^2\frac{1}{c\vert \bq\vert +\vert\omega\vert}\d\bq\\
     &\geq\frac{g^2\rho_0 R^d}{\Omega+\vert\omega\vert}\frac{1}{(2\pi)^d}\int_{\vert q_i\vert R \leq \pi/2}\prod_{i=1}^d\left(\frac{\sin(q_i R)}{q_i R}\right)^2\frac{1}{\tilde c\vert\bq\vert_{\infty} +\vert\omega\vert}\d\bq\\
     &\geq\frac{g^2\rho_0 }{\Omega+\vert\omega\vert}\frac{1}{\pi \tilde c/(2R)+\vert\omega\vert}\frac{1}{(2\pi)^d}\prod_{i=1}^d\int_{\vert q_i\vert R\leq \pi/2}\left(\frac{\sin(q_i R)}{q_i R}\right)^2R\d q_i\\
     &\geq D_d\frac{g^2\rho_0}{(\Omega+\vert\omega\vert)(\pi \tilde c/(2R)+\vert\omega\vert)} .
\end{split}
\end{align}
By taking $g^2\rho_0$ sufficiently large, the right hand side can be made strictly greater than one. Hence,  $ K_{\omega}[\rho]$ has an eigenvalue in the interval $[1,\infty)$. It follows from Theorem \ref{BirmanSchwinger2} that $H(\omega_0)\psi=\omega_0\psi$ has a non-trivial eigenpair with $\omega_0\in (-\infty,\omega]$.
\end{proof}
\subsection{An upper bound on the number of bound states}

In this section, we derive an upper bound on the number of bound states of the operator $H$.
\begin{theorem} \label{thm:NBS}
    Let $\rho\in L^d$ with compact support and $\rho\geq 0$. Then for each $d\geq 2$, there is a constant $K_d$ independent of $g$ and $\Omega$ such that $N(\rho)$, the number of bound states of $H$ with negative frequency, satisfies 
    \begin{align} \label{eq:NBS}
        N(\rho) \leq K_d\left(\frac{g^2}{\Omega c}\right)^d\int_{\R^d}\left[\rho(\bx)\right]^d \d \bx.
    \end{align}
\end{theorem}
The proof follows \cite{Daubechies1983}, \cite{Lieb1976}, and \cite{Lieb2005} quite closely. A detailed argument is included in Appendix \ref{app:WeylsLaw}.
\section{{Bound states and resonances for high contrast inclusions}} \label{sec:incl}
Throughout this section, we assume that $\rho$ is piecewise constant and supported on a connected domain $D\subset\R^d$. Specifically, we assume that
\begin{equation}\label{eq:scale}
\rho(\bx) = \rho_0\chi_{D}(\bx),
\end{equation}
where $\rho_0$ denotes the constant density inside $D$. In the subsequent sections, we will characterize the resonances and bound states in the case of \emph{high-contrast} inclusions, meaning that $D$ is a small inclusion with sufficiently high density $\rho_0$. We begin, however, by formulating the problem and stating some general properties in \Cref{sec:int}

\subsection{Integral equation formulation of the resonance problem}\label{sec:int}
As before, we seek $\omega\in\mathbb{C}$ and $\psi\neq 0$,
	with appropriate behavior as $|\bx|\to\infty$, such that 
	\begin{equation}\label{eq:ev-eqn}
		\left((-\Delta)^{1/2} - \frac{\omega}{c}\right) \psi = - \frac{g^2}{c(\omega-\Omega)}\rho(\bx)\psi.
	\end{equation}
	We will focus on the \emph{resonance problem}, corresponding to non-trivial solutions of \eqref{eq:ev-eqn} $(\omega,\psi)$ with $\Re(\omega) > 0$ and $\psi\in L^2_{\rm loc}$ which are \emph{outgoing} in the sense that $\psi$ admits an integral representation
\begin{equation}\label{eq:psi}
	\psi(\bx) = - \frac{g^2\rho_0}{c(\omega-\Omega)}\int_D G^{\omega/c}(\bx-\by)\psi(\by)\d \by,
\end{equation}
where $G^k$ is the \emph{outgoing} Green's function of $(-\Delta)^{1/2} - k$ given explicitly by
\begin{equation}\label{eq:Goutmain}
	G^k(\bx) = \begin{cases} \ds 
		\frac{1}{2\pi}\left(e^{\iu k |x| }E_1(\iu k|x|) + e^{-\iu k |x| }E_1(-\iu k|x|) \right) + \iu e^{\iu k|x|},  & d= 1,\\[0.5em] \ds
		
		\frac{1}{2\pi |\bx|} - \frac{ k}{4}\mathbf{K}_0(k|\bx|) + \frac{\iu k}{2}H_0^{(1)}(k|\bx|),  & d= 2,\\[0.7em]  
		
		\frac{1}{2\pi^2|\bx|^2}-\frac{\iu k}{4\pi^2|\bx|}\left(e^{\iu k |\bx| }E_1(\iu k|\bx|) - e^{-\iu k |\bx| }E_1(-\iu k|\bx|)\right) +  \frac{ke^{\iu k |\bx|}}{2\pi|\bx|},  & d= 3.
	\end{cases}
\end{equation}
Here, $E_1$ is the exponential integral function, $\mathbf{K}_0$ is the zeroth-order Struve function, and $H_0^{(1)}$ is the zeroth-order Hankel function of the first kind. We refer to \Cref{app:G} for the derivation and properties of $G^k$. In particular, when $k>0$ is real, $G^{k}$ is oscillatory as $|\bx|\to \infty$ and obeys the outgoing Sommerfeld radiation condition (see \Cref{sec:rad})
\begin{equation}\lim_{|\bx|\to \infty} |\bx|^{\frac{d-1}{2}}\left(\frac{\p G^k}{\p |\bx|} - \iu k G^k \right) = 0.\end{equation}

We define the integral operator $\G^\omega: L^2(D) \to L^2(D)$ as
\begin{equation}\label{eq:A}
	\G^\omega[\phi](\bx) = -(\omega-\Omega)\phi(\bx) - \frac{g^2\rho_0}{c}\int_D G^{\omega/c}(\bx-\by)\phi(\by)\d \by.\end{equation}
The following result establishes an equivalence
 between solutions of the resonance eigenvalue problem 
  and elements of the kernel of $\G^\omega$.
 \begin{prop}
     Assume that $(\psi,\omega)$ with $\psi\in L^2(D)$ 
     and with $\omega\ne\Omega$ is a solution of 
     \begin{equation}\label{eq:eval}
	\G^\omega[\psi] = 0,\quad \psi\in L^2(D).
\end{equation}
Then $\psi$, extended to all $\mathbb{R}^d$ by \eqref{eq:psi}, solves \eqref{eq:ev-eqn}. Conversely, every non-trivial solution
 of the resonance problem \eqref{eq:psi} gives rise to a non-trivial solution $\psi\in L^2(D)$ of \eqref{eq:eval}.
 \end{prop}
Suppose $(\omega,\psi)$ is a non-trivial pair for which \eqref{eq:eval} holds. If $\Im(\omega) <0$, then the extended eigenfunction $\psi$ on $\R^d$ will not belong to $L^2(\R^d)$, and $\omega$ is not an $L^2$-eigenvalue of the system \eqref{eq:system1}.

The following result, a consequence of self-adjointness of $H$, defined in  \eqref{eq:H}, constrains the imaginary parts of resonances to lying in the lower half plane.
\begin{prop}\label{prop:im}
	Assume that $\psi \in L^2(D)$ satisfies \eqref{eq:eval}. Then $\Im(\omega) \leq 0$.
\end{prop} 
\begin{proof}
	Suppose $\Re(\omega) > 0$ and $\psi\in L^2_{\rm loc}(\mathbb{R}^d)$
	satisfies \eqref{eq:ev-eqn}, and \eqref{eq:psi}. For $\Im(\omega)>0$, we know from \Cref{sec:decay} that the outgoing Green's function $G$ satisfies $G^\omega(\bx) = O(|\bx|^{-(d+1)})$ as $|\bx|\to \infty$.
 Hence, $\psi\in L^2(\mathbb{R}^d)$ and 
	$(\psi,\phi)$, where $\phi= g(\omega-\Omega)^{-1}\sqrt{\rho}\psi$, is bound state.
	Since $\Im(\omega) \neq 0$, this contradicts self-adjointness of $H$ as an operator on $L^2(\R^d) \times L^2(\R^d)$, defined in \eqref{eq:H}.
\end{proof}

\begin{remark}\label{rmk:neg}
	The problem of bound states ($\Re(\omega)<0$ and $\Im(\omega) = 0$) can be treated analogously. If we seek solutions with $\omega$ real and negative, we have an analogous integral representation \eqref{eq:psi}, where the Green's function is now rapidly decaying as $|\bx| \to \infty$ and the corresponding solution $\psi$ lies in $L^2(\R^d)$. We will present the analysis for the resonant problem only, and outline the main differences for the bound states in the remarks following \Cref{thm:2D}. Furthermore, using the appropriate Green's function as given in \Cref{app:G}, \Cref{prop:im} can easily be extended to conclude that there are no resonances with negative real part ($\Re(\omega)<0$ and $\Im(\omega)\neq 0$).
\end{remark}
\subsection{Resonances and bound states in two and three dimensions}
We now turn to the high-contrast problem. We will derive asymptotic expansions of the nonlinear eigenvalues in the limit when the size of $D$ tends to zero. We focus on the resonant case $\Re(\omega) > 0$.

We assume that $D=B_\epsilon$ is a rescaling of some connected domain $B_1\subset \R^d$,
\begin{equation}\label{eq:D}
	B_\epsilon = \epsilon B_1, \qquad 	\rho(\bx) = \rho_0(\epsilon)\chi_{B_\epsilon}(\bx),
\end{equation}
for $\epsilon>0$. To highlight the dependence on $\epsilon$, we will write $\G^{\omega}_\epsilon$ for $\G^\omega$ as defined in \eqref{eq:A}, with $D=B_\epsilon$ as in \eqref{eq:D}.
We additionally assume that $\rho_0$ depends on $\epsilon$. When $d\in \{2,3\}$, we assume that $\rho$ scales as 
\begin{equation}\rho_0(\epsilon) = \frac{s_0}{\epsilon},\end{equation}
as $\epsilon \to 0$, for some constant $s_0 >0$. As we shall see, this scaling of $\rho_0$ makes the resonances of order $O(1)$ as $\epsilon \to 0$. We define the map $S_\epsilon$ as
\begin{equation}S_\epsilon: L^2({B_\epsilon}) \to L^2(B_1), \quad S_\epsilon f(\bx) = f(\epsilon\bx).\end{equation}
Using the expansion of $G^k$ in \Cref{sec:sing}, we obtain an expansion of $\epsilon^dG^k(\epsilon \bx)$ given by
\begin{equation}\label{eq:Gexp23D}
	\epsilon^dG^k(\epsilon \bx) = \sum_{n=0}^\infty \epsilon^{n+1}A_n^k(\bx) +  \sum_{n=d-1}^\infty \epsilon^{n+1}\log(\epsilon)\Glog{n}{k}(\bx),
\end{equation}
where the functions $A_n^k$ and $\Glog{n}{k}$ can be explicitly computed (the first few terms, which will be relevant in this section, are given in \Cref{sec:sing}). Crucially, for $n\geq d$, these functions are uniformly bounded for $\bx \in B_1$, so \eqref{eq:Gexp23D} converges uniformly for $\bx \in B_1$ and $0<\epsilon\ll 1$. We define the operators $\A_n^\omega: L^2(B_1) \to L^2(B_1)$ and $\Alog{n}{\omega}: L^2(B_1) \to L^2(B_1)$, for $n>0$, as 
\begin{equation}\A_n^\omega[\phi](\bx) = -\frac{g^2s_0}{c}\int_{B_1} A^{\omega/c}_n(\bx-\by)\phi(\by) \d \by, \quad \Alog{n}{\omega}[\phi](\bx) = -\frac{g^2s_0}{c}\int_{B_1}\Glog{n}{\omega/c}(\bx-\by)\phi(\by) \d \by.\end{equation}
We now define $\A_0^\omega: L^2(B_1) \to L^2(B_1)$ by
\begin{equation} \label{eq:A0}
\A_0^\omega[\phi](\bx) = -(\omega-\Omega)\phi(\bx) -\frac{g^2s_0}{c}\int_{B_1} A_0(\bx-\by)\phi(\by) \d \by.
\end{equation}
Here $A_0 = A_0^k$ is independent of $k$, so we omit the superscript. Throughout this work, we define $\L(X)$ to be the space of bounded, linear operators on the Banach space 
$X$, equipped with the operator norm. We then have the following result, which follows from \eqref{eq:Gexp23D} along with \eqref{eq:bound} and a change of variables $\by=\eps \bx$. 
\begin{prop}
	Assume $B_\epsilon = \epsilon B_1$ and $\rho(\bx) =  \frac{s_0}{\epsilon}\chi_{B_\epsilon}(\bx)$. We then have 
	\begin{equation}S_\epsilon\G_\epsilon^\omega S_\epsilon^{-1} = \sum_{n=0}^\infty \epsilon^n\A_n^\omega + \sum_{n=d-1}^\infty \epsilon^n\log(\epsilon)\Alog{n}{\omega},
	\end{equation} 
where the convergence holds in $\L\bigl(L^2(B_1)\bigr)$.
\end{prop}

$\A_0^\omega$ can be seen as the limiting operator of $\G_\epsilon^\omega$ as $\epsilon \to 0$. Moreover, the eigenvalue problem $\A_0^\omega \psi = 0$
amounts to a linear eigenvalue problem. We define the integral operator $L_0:L^2(B_1)\to L^2(B_1)$ by
\begin{equation} \label{eq:L0}
	L_0\phi(\bx)=\frac{g^2s_0}{c}\int_{B_1} A_0(\bx-\by)\phi(\by)\d \by,
\end{equation}
then $\A_0^\omega \psi = 0$ is equivalent to 
\begin{equation}L_0 \psi = (\Omega-\omega)\psi,\end{equation}
where the parameter $\Omega$ leads to a shift of the eigenvalues.

From \eqref{eq:Gn2D} and \eqref{eq:Gn3D}, we have that $A_0$ is strictly positive and weakly singular at $\bx = 0$. We then have that $L_0$ is a positive, compact, and self-adjoint operator on $L^2(B_1)$, which means that it has a sequence of real, positive, eigenvalues converging to zero. Correspondingly, $\A_0^\omega$ has a sequence of eigenvalues $\omega_j, \ j=1,2,...$ converging to $\Omega$ from below. Next, we will compute the eigenvalues $\omega = \omega_\epsilon$ of $\G_\epsilon^\omega$ which arise as perturbations $\omega_j$. For simplicity, we will focus on the case where $\omega_j$ is simple. As we show below, this is at least true for the smallest eigenvalue. 
\begin{prop}
    Let $\A_0^\omega$ be given by \eqref{eq:A0}. Then the lowest eigenvalue is simple and the associated eigenfunction can be chosen to be strictly positive.
\end{prop}
\begin{proof}
	Recall that $\A_0^\omega\phi = 0$ is equivalent to $L_0\phi = (\Omega-\omega)\phi$, where $L_0$ is given by \eqref{eq:L0}. This operator is positive, compact and self-adjoint, so that $\|L_0\|$ is its largest eigenvalue. Moreover, the kernel of $L_0$ is strictly positive by \eqref{eq:Gn2D} (in $d=2$) or \eqref{eq:Gn3D} (in $d=3$), and thus it is positivity-improving in the sense that $L_0\phi$ is a positive function whenever $\phi\neq 0$ is a nonnegative function. Thus by Theorem XIII.43 of ~\cite{ReedSimonIV}, we have that this eigenvalue is simple and has a strictly positive eigenfunction $\phi_0 > 0$. In terms of the operator $\A_0^\omega$ this means that
\begin{align}
    \A_0^{\Omega-\|L_0\|}\phi_0 = 0,
\end{align}
where the lowest eigenvalue is $\omega_0 = \Omega - \|L_0\|$.
\end{proof}

We observe that $\G_\epsilon^\omega$ is a holomorphic operator-valued function of $\omega$ in the right half-plane (see \Cref{sec:GS}). The following proposition, whose proof is given in \Cref{sec:GS}, follows from the Gohberg-Sigal perturbation theory for holomorphic operators.
\begin{prop} \label{prop:pert}
	Let $d\in\{2,3\}$ and $\omega_j\in \R\setminus\{0,\Omega\}$ be a simple eigenvalue of $\A^{\omega}_{0}$. Then, for small enough $\epsilon$, there exists a simple eigenvalue $\omega = \omega(\epsilon) \in \C$ of $\G_\epsilon^\omega$, continuous as a function of $\epsilon$ and satisfying $\lim_{\epsilon \to 0}\omega(\epsilon) = \omega_j$.
\end{prop}

We begin with the three-dimensional case.
\begin{theorem}\label{thm:3D}
	Let $d=3$ and $\omega_j>0$ be a simple eigenvalue of $\A_0^\omega$, corresponding to the normalized eigenfunction $\psi_j$. Then, for small enough $\epsilon$, there is a simple eigenvalue $\omega(\epsilon)$ of $\G_\epsilon^\omega$ satisfying
	\begin{align}
		\Re\bigl(\omega(\epsilon) \bigr) &= \omega_j + \epsilon\left\langle \psi_j, \A_1^{\omega_j}\psi_j \right\rangle + O(\epsilon^2), \label{eq:3dRe} \\
		\Im\bigl(\omega(\epsilon) \bigr) &= -\epsilon^2\frac{\omega_j^2g^2s_0}{2\pi c^3}\left(\int_{B_1} \psi_j(\bx)\d\bx\right)^2 + O(\epsilon^3\log\epsilon), \label{eq:3dIm}
	\end{align}
	where $\langle \cdot, \cdot \rangle$ denotes the $L^2({B_1})$-inner product.
\end{theorem}
\begin{proof}
	First, since $\G_\epsilon^\omega$ is holomorphic for $\omega\in U$, we know that we can find a neighbourhood $V$ of $\omega_j$ such that, for small enough $\epsilon$, there is a single, simple eigenvalue $\omega = \omega(\epsilon)$ of $\G_\epsilon^\omega$ in $V$. We then have a complete asymptotic expansion of $\omega(\epsilon)$ in terms of $\epsilon$  (see \Cref{sec:GS}) of the form
	\begin{equation}\label{eq:gat}
		\omega(\epsilon)  - \omega_j = \frac{1}{2\pi \iu}\sum_{p=1}^\infty\frac{1}{p}
		\tr\int_{\partial V}\left[\left(\A_0^w\right)^{-1}\left(\A_0^w - S_\epsilon\G_\epsilon^{w}S_\epsilon^{-1} \right)\right]^p\d w.
	\end{equation}
	We can easily compute a Laurent expansion of $\left(\A_0^\omega\right)^{-1}$. First, we observe that 
	\begin{equation}\frac{\d}{\d \omega} \A_0^\omega = -I.\end{equation}
	Since $\omega_j$ is a simple pole of $\left(\A_0^\omega\right)^{-1}$ we have, for $\omega$ in a neighbourhood of $\omega_j$, that 
	\begin{equation}\left(\A_0^\omega\right)^{-1} = -\frac{\langle \psi_j, \cdot \rangle\psi_j}{\omega-\omega_j} + \mathcal{R}(\omega),\end{equation}
	where $\mathcal{R}$ is holomorphic in a neighbourhood of $\omega_j$. Moreover, $\A_0$ maps real-valued functions to real-valued functions, so we can take $\psi_j$ and hence $\mathcal{R}(\omega_j)\psi_j$, to be real-valued. If we consider \eqref{eq:gat} to first order, we have
	\begin{align}
		\omega(\epsilon)  - \omega_j &= -\frac{\epsilon}{2\pi \iu}
		\tr\int_{\partial V}\left(\A_0^w\right)^{-1}\A_1^{w}\d w + O(\epsilon^2) \nonumber \\
		&=\epsilon\left\langle \psi_j, \A_1^{\omega_j}\psi_j \right\rangle + O(\epsilon^2). \label{eq:first}
	\end{align}
	Eq.~\eqref{eq:first} gives the leading-order approximation of $\omega(\epsilon)$ in the limit $\epsilon \to 0$.
	Since $\psi_j$ is real-valued, it is clear that $\left\langle \psi_j, \A_1^{\omega_j}\psi_j \right\rangle\in \R$, which proves \eqref{eq:3dRe}. Taking the imaginary part, we have from \eqref{eq:gat} that
	\begin{align}
		\Im\bigl(\omega(\epsilon) \bigr)  &=\Im\left( -\frac{\epsilon^2}{2\pi \iu}
		\int_{\partial V} \left(-\frac{\langle \psi_j, \A_2^{w}\psi_j \rangle}{w-\omega_j}\right)\d w \right) + O(\epsilon^3\log\epsilon) \nonumber \\
		&= \epsilon^2\Im\bigl( \langle \psi_j, \A_2^{\omega_j}\psi_j \rangle \bigr) + O(\epsilon^3\log\epsilon)\nonumber \\
		&= -\epsilon^2\frac{\omega_j^2g^2s_0}{2\pi c^3}\left(\int_{B_1} \psi_j(\bx)\d\bx\right)^2 + O(\epsilon^3\log\epsilon).
	\end{align}
\end{proof}
In two dimensions, we have a similar theorem, which can be proved analogously.
\begin{theorem}\label{thm:2D}
	Let $d=2$ and $\omega_j>0$ be a simple eigenvalue of $\A_0^\omega$, corresponding to the normalized eigenfunction $\psi_j$. Then, for small enough $\epsilon$, there is a simple eigenvalue $\omega$ of $\G_\epsilon^\omega$ satisfying
	\begin{align}
		\Re\bigl(\omega(\epsilon) \bigr) &= \omega_j + \epsilon\log(\epsilon)\frac{\omega_j g^2s_0}{2\pi c^2}\left(\int_{B_1} \psi_j(\bx)\d\bx\right)^2 + O(\epsilon), \label{eq:2dRe}\\
		\Im\bigl(\omega(\epsilon) \bigr)  &=  -\epsilon\frac{\omega_jg^2s_0}{2c^2}\left(\int_{B_1} \psi_j(\bx)\d\bx\right)^2+ O(\epsilon^2\log\epsilon), \label{eq:2dIm}
	\end{align}
	where $\langle \cdot, \cdot \rangle$ denotes the $L^2({B_1})$-inner product.
\end{theorem}
Next we make several remarks about the theorems above.
\begin{enumerate}
    \item (Agreement with \Cref{prop:im}) In agreement with \Cref{prop:im}, we also observe that the imaginary parts in \eqref{eq:3dIm} and \eqref{eq:2dIm} are nonpositive for small enough $\epsilon$.
    \item (Relation to other work) In \cite{hoskins_2021}, the dynamics of a related model, with a finite number of discretely supported atoms, is studied. In contrast to the present case, this model has a finite number of eigenvalues. We expect the accumulation of eigenvalues to give rise to different physical behavior. 
    \item (Application to multiple resonances) The method used here can also be applied to multiple eigenvalues of $\A_0^\omega$, which might split into simple eigenvalues of $\G_\epsilon^\omega$ for nonzero $\epsilon$ \cite{ammari2004splitting}.
    \item (Existence of eigenvalues) Changing $\Omega$ corresponds to a translation of the eigenvalues of the limiting operator $\A_0^\omega$. In particular, if $\Omega < \|L_0\|$, there is at least one negative eigenvalue, corresponding to a bound state. We can use a similar analysis to compute  asymptotic expansions of these negative eigenvalues. In this case, we use the Green's function for negative $\omega$ when defining the operator $\G_\epsilon^\omega$ in \eqref{eq:A}. The limiting operator $\A_0^\omega$ in this case is identical to the resonant case in \eqref{eq:A0}. Following the same method, we can prove analogues of \Cref{thm:3D} and \Cref{thm:2D}. If $\omega_j < 0$ is a simple eigenvalue of $\A_0^\omega$ then for small $\epsilon$ there is a simple eigenvalue $\omega<0$ of $\G_\epsilon^\omega$ whose real part is given by \eqref{eq:3dRe} or \eqref{eq:2dRe} and whose imaginary part vanishes.
    \item (Agreement with \Cref{thm:boundstates}) The condition $\Omega < \|L_0\|$ agrees with \eqref{eq:Omegabound} of \Cref{thm:boundstates}. Observe that $\|L_0\| = \frac{g^2\rho_0(\epsilon)}{c}\tilde K$ for some $\tilde K \neq 0$ independent of $g,\rho_0,\Omega$ and $\epsilon$. Then, if 
	$$\frac{g^2\rho_0(\epsilon)}{\Omega c} > \tilde{K}^{-1},$$
	there is at least one bound state. We note, however, that the asymptotic analysis can only guarantee the existence of this bound state for sufficiently small $\epsilon$, whereas the analysis of \Cref{sec:bound} holds independent of the size of $\epsilon$.
    \item (Generating bound states) By decreasing $\Omega>0$, we can increase the number of negative eigenvalues of $\A_0^\omega$. Unless there are only finitely many eigenvalues $\omega_j\neq \Omega$, we can create arbitrarily many bound states of the full operator $\G_\epsilon^\omega$ by choosing $\Omega$ and $\epsilon$ small enough. This is consistent with \Cref{thm:NBS}. For small $\Omega$ the right-hand side of \eqref{eq:NBS} is unbounded.
\end{enumerate}

\begin{figure}
	\centering
	\includegraphics[width=0.6\linewidth]{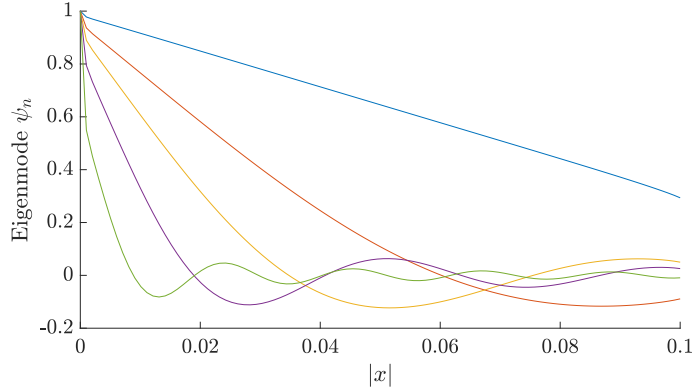}
	\caption{First five eigenmodes of $\G_\epsilon^{\omega/c}$ in the case $\epsilon = 0.1$ for the unit sphere. The results are computed by Nyström discretization of the integral equation \eqref{eq:psi}, as outlined in \Cref{sec:num}.}\label{fig:mult}
\end{figure}


\begin{figure}
	\centering
	\begin{subfigure}[t]{0.45\linewidth}
		\includegraphics[width=\linewidth]{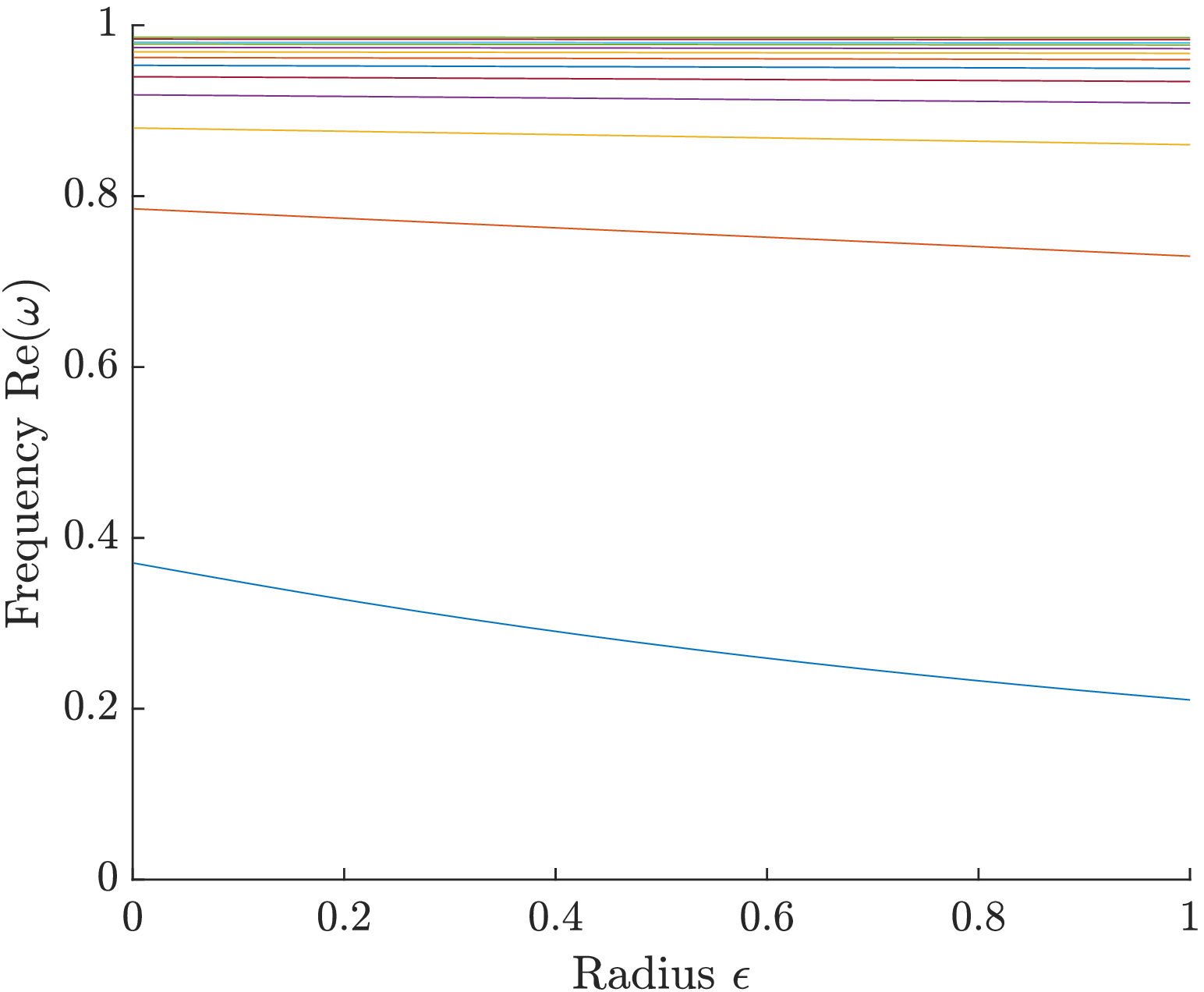}
		\caption{Real part.}
	\end{subfigure}
	\hspace{0.5cm}
	\begin{subfigure}[t]{0.45\linewidth}
		\includegraphics[width=\linewidth]{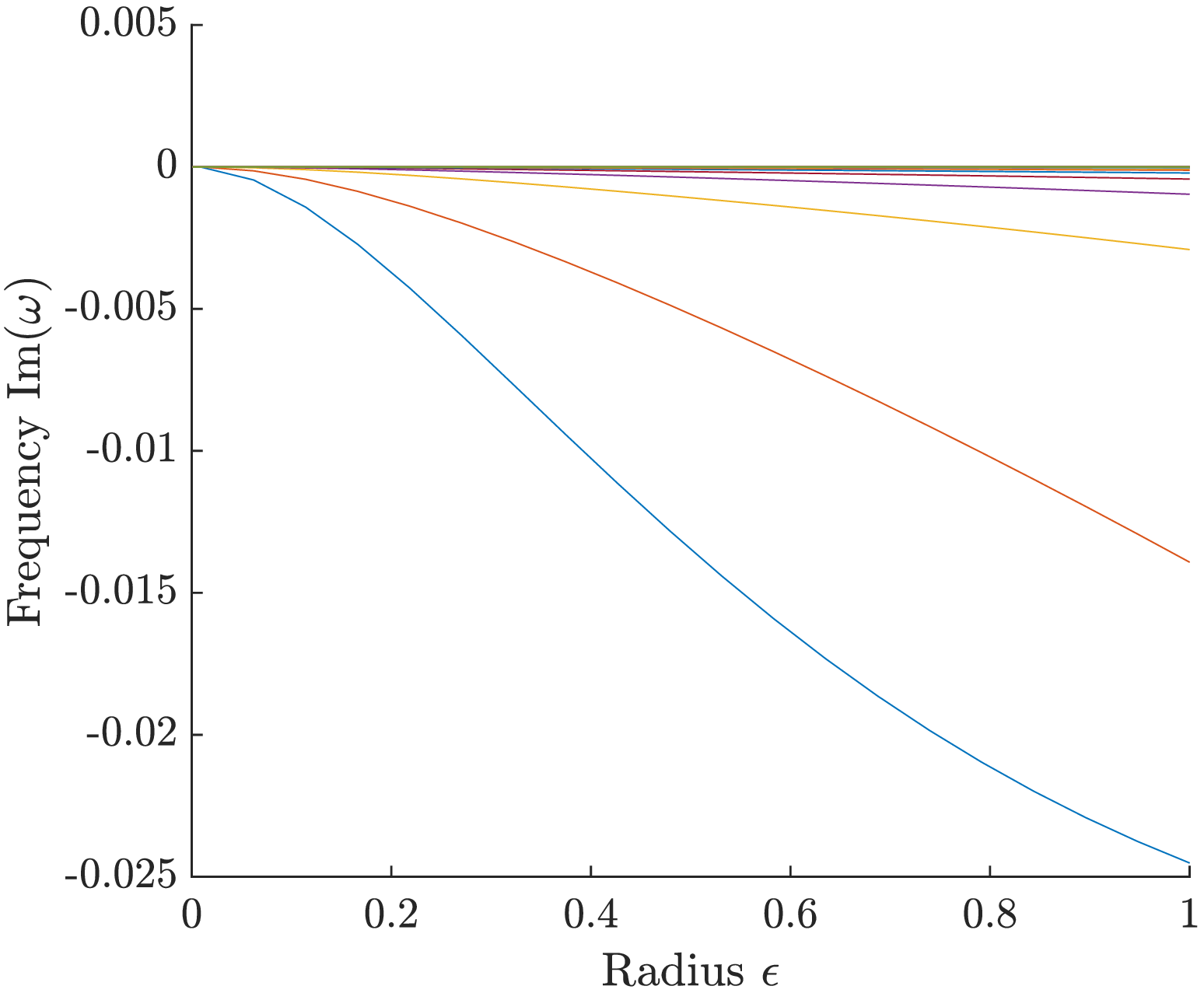}
		\caption{Imaginary part.}
	\end{subfigure}
	\caption{Computation of the resonant frequencies as function of the radius $\epsilon$ of the inclusion, in the case of spherical inclusions in three dimensions. Here $\Omega =1$ corresponds to the accumulation point of the resonances. These are computed through the Nyström method, as outlined in \Cref{sec:num}.}\label{fig:mult_freq}
\end{figure}

\subsection{Approximation of the lowest eigenvalue}\label{sec:approx}
Since $\epsilon$ is small, we can formally estimate the first eigenvalue by approximating $\psi(\bx)$ on ${B_\epsilon}$ by its value   at $\bx=0$. That is, $\psi(\bx) \approx\psi(0)$  for $\bx \in {B_\epsilon}$. For simplicity, we restrict our attention to the case of a sphere. In the current setting, we have from \eqref{eq:Gexp} that
\begin{equation}
	\int_{B_\epsilon} G^{\omega/c}(\bx,\by)\psi(\by)\d \by \approx \psi(0)\left(
	\frac{2\epsilon}{\pi} + \frac{\epsilon^2\omega}{c} +\iu \frac{2\epsilon^3\omega^2}{3c^2}\right). 
\end{equation}
If we evaluate \eqref{eq:psi} at $\bx = 0$, this yields
\begin{equation}\psi(0) \approx-\frac{g^2s_0\psi(0)}{c(\omega-\Omega)}\left(
\frac{2}{\pi} + \frac{\epsilon\omega}{c} +\iu \frac{2\epsilon^2\omega^2}{3c^2}\right).
\end{equation}
In order to have nonzero solutions, we require $\psi(0) \neq 0$. It therefore follows that
\begin{equation}\omega-\Omega \approx-\frac{g^2s_0}{c}\left(
\frac{2}{\pi} + \frac{\epsilon\omega}{c} +\iu \frac{2\epsilon^2\omega^2}{3c^2}\right).
\end{equation}
We now define $\alpha = \frac{2g^2s_0}{\pi c}$. We then have that 
\begin{align}
	\Re\bigl(\omega(\epsilon) \bigr)  &\approx  \Omega - \alpha - \frac{\epsilon(\Omega -\alpha)}{c}, \\
	\Im\bigl(\omega(\epsilon) \bigr)  &\approx  -\epsilon^2\alpha\pi\frac{(\Omega -\alpha)^2}{3c^2}.
\end{align}

\begin{figure}
	\centering
	\begin{subfigure}[t]{0.45\linewidth}
		\includegraphics[width=\linewidth]{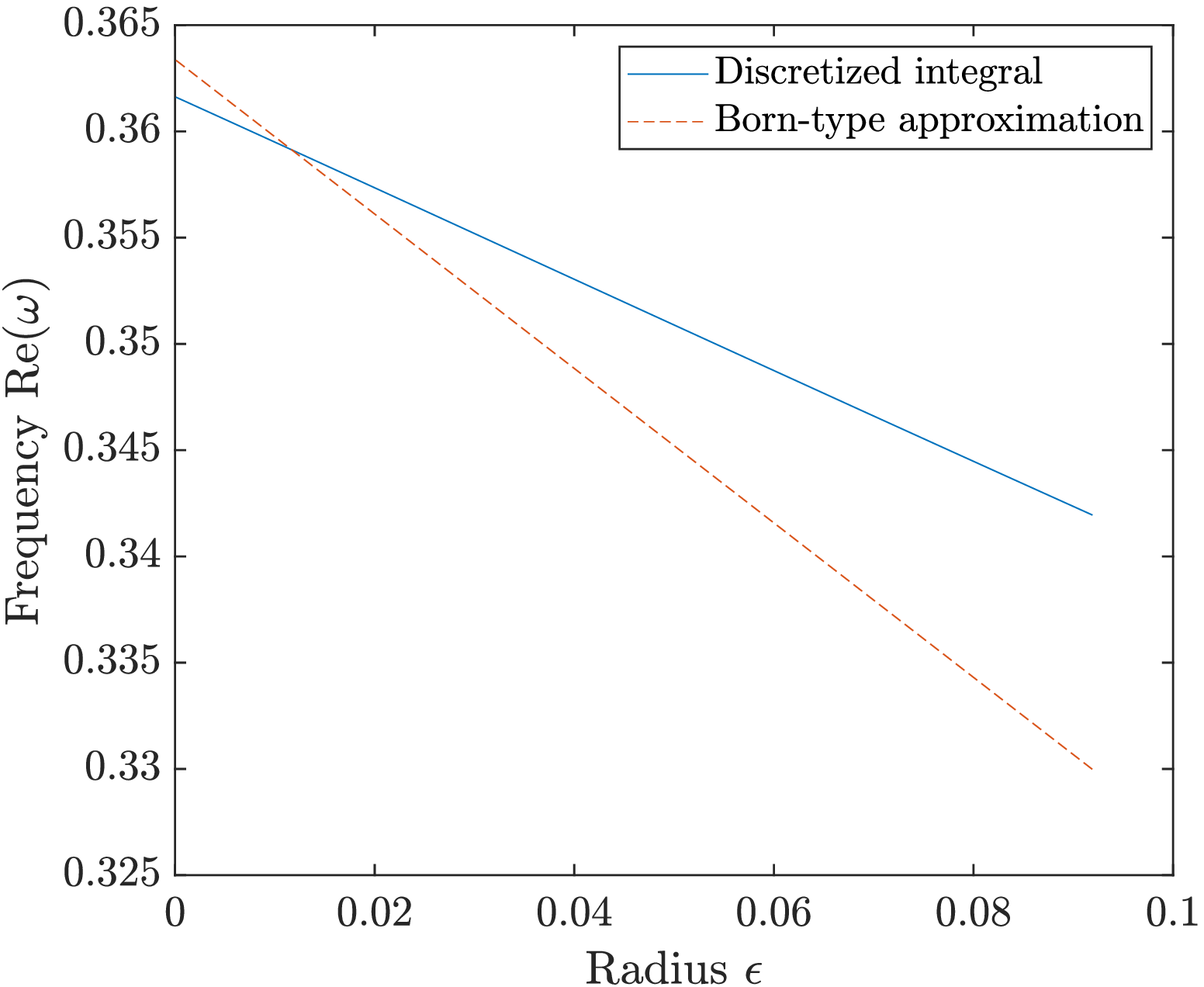}
		\caption{Real part.}
	\end{subfigure}
	\hspace{0.5cm}
	\begin{subfigure}[t]{0.45\linewidth}
		\includegraphics[width=\linewidth]{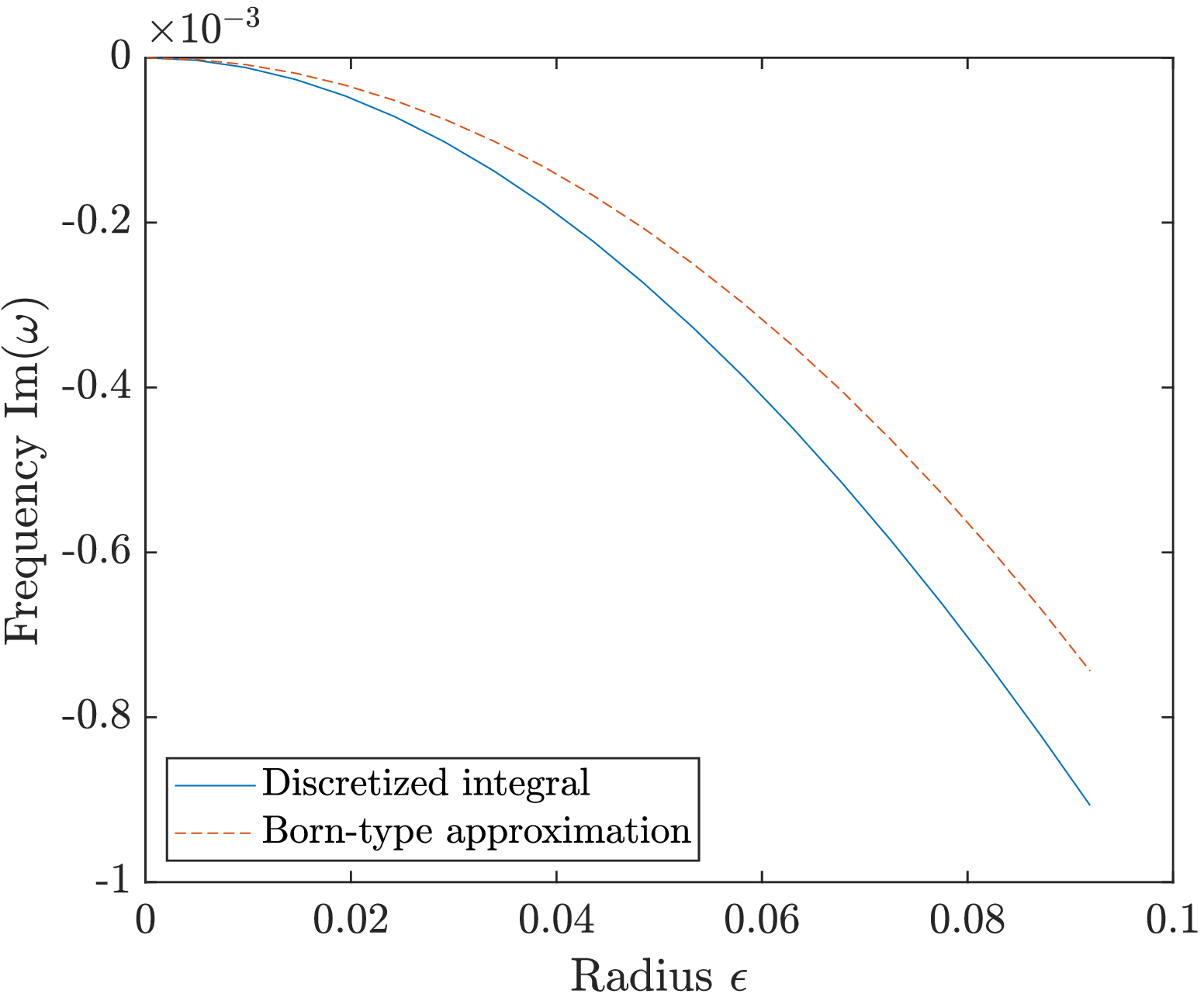}
		\caption{Imaginary part.}
	\end{subfigure}
	\caption{Comparison between the first eigenvalue computed either by discretizing the integral equation \eqref{eq:psi}  or through the approximation made in \Cref{sec:approx}.}\label{fig:born}
\end{figure}

\subsection{Resonances and bound states in one dimension}
In one dimension, we must choose a different scaling of $\rho$ in order to get resonances of order $O(1)$ as $\epsilon \to 0$. We assume that
\begin{equation}\label{eq:scale1D}
	{B_\epsilon} = \epsilon {B_1}, \qquad 	\rho_0(\epsilon) = -\frac{s_0}{\epsilon\log\epsilon}\chi_{\epsilon {B_1}},
\end{equation}
as $\epsilon \to 0$, where the domain ${B_1}\subset \R$ and $s_0>0$. We then have 
\begin{equation}S_\epsilon\G_\epsilon^\omega S_\epsilon^{-1} = \sum_{n=0}^\infty \epsilon^n\Alog{n}{\omega} + \sum_{n=0}^\infty \frac{\epsilon^{n}}{\log(\epsilon)}\A_{n}^\omega,\end{equation}
where the convergence holds in $\L\bigl(L^2({B_1})\bigr)$. Observe that the limiting operator is now given by 
\begin{equation}\Alog{0}{\omega}[\phi](x) = -(\omega-\Omega)\phi(x) -\frac{g^2s_0}{\pi c}\int_{{B_1}}\phi(y) \d y.\end{equation}
If $\omega \neq \Omega$, this is a rank-1 perturbation of the identity and has a single eigenvalue. Based on this observation, we can prove the following result. 
\begin{theorem}\label{thm:1Dlog}
	Assume that $\rho(\bx) = -\frac{s_0}{\epsilon\log\epsilon}\chi_{B_\epsilon}(\bx)$ and that $ \Omega - \frac{g^2s_0|{B_1}|}{\pi c}>0$. Then for $\epsilon \to 0$, there is an eigenvalue $\omega = \omega(\epsilon)$ of $\G_\epsilon^\omega$ satisfying
	\begin{align*}
		\Re\bigl(\omega(\epsilon) \bigr)  &= \Omega - \frac{g^2s_0|{B_1}|}{\pi c} + O\left(\frac1{\log\epsilon}\right), \\ 
		\Im\bigl(\omega(\epsilon) \bigr)  &= \frac{g^2s_0|{B_1}|}{c\log\epsilon} + O(\epsilon).
	\end{align*}
\end{theorem}
\begin{remark} \label{rmk:1dbound}
	We know from \Cref{thm:boundstates} that there will always be an eigenvalue $\omega <0$. If $\Omega -\frac{g^2s_0|{B_1}|}{\pi c} < 0$, we can follow them proof of \Cref{thm:3D} to show that there is a negative eigenvalue whose real part is given  in \Cref{thm:1Dlog} and whose imaginary part vanishes. If, however, $\Omega -\frac{g^2s_0|{B_1}|}{\pi c} > 0$, the resonance described in \Cref{thm:1Dlog} has positive real part for small $\epsilon$. In the case $d=1$, the Green's function $G^{\omega}$ diverges as $\omega\to 0$. Formally we have the limiting problem
	\begin{equation}-(\omega-\Omega)\phi(x) -\frac{g^2s_0 \log \left( \frac{-\epsilon\omega}{c}\right) }{\pi c\log \epsilon}\int_{{B_1}}\phi(y) \d y = 0,\end{equation}
	which has an eigenvalue $\omega = \omega(\epsilon)$ such that $\omega < 0$ and $\omega \to 0 $ as $\epsilon\to 0$. Asymptotically, this eigenvalue is given by
	\begin{equation} \label{eq:boundfa}
		\omega(\epsilon) = -c\epsilon^{p}, \qquad p = \frac{\Omega\pi c}{g^2s_0|{B_1}|} - 1>0.
	\end{equation}
	However, since $G^\omega(x)$ is not holomorphic for $\omega$ in a neighbourhood of $0$, the Gohberg-Sigal eigenvalue perturbation method does not apply around this point. Nevertheless, as shown in \Cref{fig:bound}, this asymptotic formula agrees well with numerical computations of the negative eigenvalue.
\end{remark}
\begin{figure}
	\centering
	\includegraphics[width=0.6\linewidth]{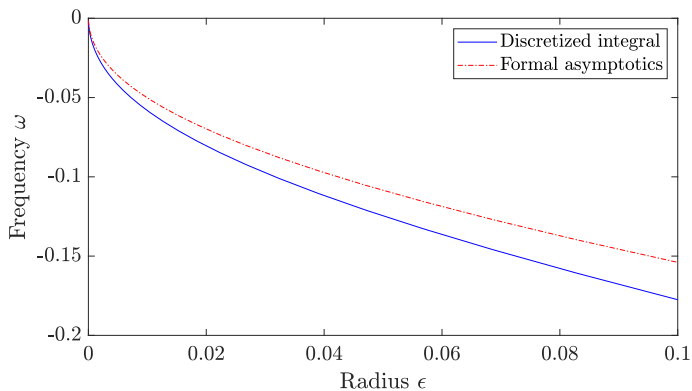}
	\caption{Numerical computation of the negative eigenvalue, corresponding to a bound state, in the case $d=1$ and $\Omega -\frac{g^2s_0|{B_1}|}{\pi c} > 0$. The value computed through the formal asymptotics \eqref{eq:boundfa} agrees well with the value computed by discretization of the integral operator.}\label{fig:bound}
\end{figure}

\subsection{Numerical computations}\label{sec:num}
The integral equation \eqref{eq:psi} can be numerically solved using the Nystrom method. We make use of spherical coordinates in three dimensions and for simplicity, we consider radially symmetric eigenfunctions. We use the trapezoidal rule to evaluate the angular integrals and Gaussian quadrature for the radial integral.  The nonlinear eigenvalue problem for the operator $\G_\epsilon^\omega: L^2(D) \to L^2(D)$,
$\G_\epsilon^\omega \psi = 0, $
is then approximated by a nonlinear eigenvalue problem for the matrix $\mathbf{G}^\omega$, which is defined by $\mathbf{G}^\omega \bx = 0$.
Numerically, we solve this equation using Muller's method for the function $f$, where
$f(\omega) = 0$ and $f(z)$ is the smallest eigenvalue of $\mathbf{G}^\omega.$

\section{Discussion}
\label{sec:conclusions}

In this paper we have considered the quantum optics of a single photon interacting with a system of two level atoms. Mathematically, this leads to a nonlinear eigenproblem for a nonlocal PDE.
 We have derived necessary and sufficient conditions for the existence of bound states, along with an upper bound on the number of such states. The upper bound on the number of bound states diverges at low atomic resonance frequencies. We have also considered the bound states and resonances for models with small high contrast atomic inclusions. In this setting, we obtain a linear eigenvalue problem with a sequence of real eigenvalues accumulating at the atomic resonance frequency. The sign of each eigenvalue dictates if it corresponds to a bound state (negative) or a resonance (positive real part). We have derived asymptotic formulas for these eigenvalues, and corroborated these formulas in numerical computations.

There are a number of questions regarding the system \eqref{eq:b4}-\eqref{eq:b5} that we intend to pursue in further work. These include:
\begin{enumerate}
    \item (Accumulation of resonances) If $\omega \neq \Omega$,  the limiting operator $\A_0^\omega$ is a compact perturbation of the identity and therefore has a sequence of eigenvalues  $\omega_j$ converging to $\Omega$ as $j\to \infty$. Moreover, since $\psi_j$ are orthogonal,  $\int_{B_1}\psi_j\d \bx$ tends to zero. For a fixed, nonzero $\epsilon$, we expect the full operator $\G_\epsilon^\omega$ to also have a sequence of eigenvalues whose real parts converge to $\Omega$ and whose imaginary parts converge to $0$. This agrees with the numerically computed values shown in \Cref{fig:mult_freq}. We observe, however, that the asymptotic expansions in \Cref{thm:3D} and \Cref{thm:2D} are formulated for a fixed $\omega_j$ and might fail to be uniform in $j$. It would be of interest to obtain bounds which are uniform in $j$ and prove that the resonances accumulate at $\Omega$.
    \item (Proof of \Cref{rmk:1dbound}) Can one provide a justification of the asymptotic formula \eqref{eq:boundfa} for the limiting problem described in \Cref{rmk:1dbound}?
\end{enumerate}

\appendix

\section{Existence and uniqueness of solutions to  \eqref{eq:b4}--\eqref{eq:b5}} 
\label{app:EandU}
We begin by making the substitution
\begin{align}
    \phi(x,t) = \sqrt{\rho(x)}a(x,t)
\end{align}
in  \eqref{eq:b4}--\eqref{eq:b5}, which becomes
\begin{align}
        \iu \partial_t \psi &= c(-\Delta)^{1/2}\psi + g\sqrt\rho \phi,\label{pde-1}
        \\
     \iu \partial_t \phi &= g\sqrt\rho \psi + \Omega \phi ,
     \label{pde-2}
\end{align}
where $\phi\equiv 0$ for $x\notin\supp{\rho}$. Consider the initial conditions
\begin{align}
    \psi(x,0) &=\psi_0(x), \\
    \phi(x,0) &=\phi_0(x),
\end{align}
where $\psi_0,\phi_0\in L^2(\R^d)$ and $\supp\phi_0\subset \supp\rho$. Applying the Duhamel formula to \eqref{pde-1} and integrating \eqref{pde-2} we obtain
\begin{align}
    \psi(x,t) &= e^{-\iu c(-\Delta)^{1/2}}\psi_0(x) -\iu g\int_0^t e^{-\iu c(-\Delta)^{1/2}(t-s)}(\sqrt\rho\phi)(x,s)\d s, \\
    \phi(x,t) &= e^{-\iu \Omega t}\phi_0(x) -\iu g\int_0^t e^{-\iu \Omega(t-s)}\sqrt{\rho}(x)\psi(x,s)\d s .
\end{align}
\begin{remark}
Note  that if $\supp \phi_0\subset\supp \rho$, then $\supp\phi \subset\supp\rho$ for all $t>0$ by the above formula. 
\end{remark}
\begin{theorem}
Let $\rho\in L^{\infty}(\R^d)$ be nonnegative with compact support. Then there is a unique solution to the above equations which belongs to $L^2(\R^d)$ for every time $t>0$.
\end{theorem}
\begin{proof}
Let $S = C_t L_x^2(I\times \R^d)$, with $I=[0,T]$ and $T>0$. We will denote the ball of radius $R$ centered at the origin in $S$ by $B_R(0)$. We define the map $\Phi:S\times S\to S\times S$ by
\begin{align}
    \Phi(\psi,\phi) &= (\Phi_1(\psi,\phi),\Phi_2(\psi,\phi)), \\
    \Phi_1(\psi,\phi) &=e^{-\iu c(-\Delta)^{1/2}}\psi_0(x) -\iu g\int_0^t e^{-\iu c(-\Delta)^{1/2}(t-s)}(\sqrt\rho \phi)(x,s)\d s, \\
    \Phi_2(\psi,\phi) &= e^{-\iu \Omega t}\phi_0(x) -\iu g\int_0^t e^{-\iu \Omega(t-s)}\sqrt{\rho}(x)\psi(x,s)\d s .
\end{align}
Suppose that $\|(\psi_0,\phi_0)\|_{(L_x^2)^2}= R/2$. We find that
\begin{align}
    \|\Phi_1(\psi,\phi)\|_{S} &\leq \|\psi_0\|_{L_x^2} +g\|\sqrt\rho\|_{L^\infty}T\|\phi\|_{S}, \\
        \|\Phi_2(\psi,\phi)\|_{S} &\leq \|\phi_0\|_{L_x^2} +g\|\sqrt\rho\|_{L^\infty}T\|\psi\|_{S} .
\end{align}
Hence
\begin{align}
    \|\Phi(\psi,\phi)\|_{S^2} \leq \|(\psi_0,\phi_0)\|_{(L_x^2)^2} +gT\|\sqrt\rho\|_{L^\infty}\|(\psi,\phi)\|_{S^2}.
\end{align}
Next we define $T$ as
\begin{align}
     T = \dfrac{1}{2g\|\sqrt\rho\|_{L^\infty}}.
\end{align}
Then since $\|(\psi_0,\phi_0)\|_{(L_x^2)^2}= R/2$, we have
\begin{align}
    \Phi: B_{R}(0)\to B_{R}(0)
\end{align}
and thus
\begin{align}
    \|\Phi(\psi,\phi)\|_{S^2} < \frac{R}{2} + \frac{R}{2} = R.
\end{align}
Moreover, $\Phi$ is a contraction on $B_{R}(0)$. To see this, suppose $(\psi,\phi),(\psi',\phi')\in S\times S$. 
It follows that
\begin{align}
    \|\Phi(\psi,\phi)-\Phi(\psi',\phi')\|_{S^2} & = g\left\|\int_0^te^{-\iu c(-\Delta)^{1/2}(t-s)}(\sqrt\rho (\phi-\phi'))(x,s)\d s\right\|_{S}\\
    &+g\left\|\int_0^t e^{-\iu \Omega(t-s)}\sqrt{\rho}(x)(\psi(x,s)-\psi'(x,s))\d s \right\|_{S}\\
    &\leq gT\|\sqrt\rho\|_{L^\infty}\left(\|\phi-\phi'\|_S+\|\psi-\psi'\|_{S}\right)\\
    &\leq \frac{1}{2}\|(\psi,\phi)-(\psi',\phi')\|_{S\times S}.
\end{align}
We conclude that there is a unique solution to the system \eqref{eq:b4}--\eqref{eq:b5}. Note that since the parameter $T$ is independent of the norms of the initial data and solution, and the mass 
\begin{align}
    M(t) &= \|(\psi(t),\phi(t))\|_{L^2\oplus L^2}\\
    &= \|e^{-\iu Ht}(\psi_0,\phi_0)\|_{L^2\oplus L^2}\\
    &= \|(\psi_0,\phi_0)\|_{L^2\oplus L^2},
\end{align}
is conserved, the solutions can be extended globally in time. 
\end{proof}

\section{Proof of Theorem  \ref{thm:NBS}}\label{app:WeylsLaw}
\begin{proof}
The idea of the proof is to take advantage of the fact that  for any $F:\R^+\to \R^+$ such that $F(x)\geq 1$ whenever $x\geq 1$, we have the following bound on the number of bound states with frequency strictly less than $\omega \leq 0$, $N_{\omega}(\rho)$:
\begin{align}
    N_{\omega}(\rho)\leq\tr(F(K_{\omega}[\rho])).
\end{align}

    We begin by expressing the kernel of the operator $e^{-t(c(-\Delta)^{1/2}-\frac{\lambda g^2}{\omega-\Omega}\rho)}$ through the generalized Feynman-Kac formula
    \begin{align}
        e^{-t(c(-\Delta)^{1/2}-\frac{\lambda g^2}{\omega-\Omega}\rho)}(\bx,\by) = \int d\mu_{\bx,\by,t}(\omega) e^{-\frac{\lambda g^2}{\vert\omega-\Omega\vert}\int_0^t \rho(u(s))\d s }.
    \end{align}
Here $d\mu_{x,y,t}$ is the measure generated by the semigroup  $e^{-t(c(-\Delta)^{1/2})}$ with $\omega$ paths pinned to $x$ at time $0$ and $y$ at time $t$. Now for $\omega\leq 0$, multiply the above by $e^{\omega t}\rho(\bx)^{1/2}\rho(\by)^{1/2}$ and integrate from $t=0$ to $t=\infty$ to obtain
\begin{align}
    A &:= \left\{\rho^{1/2}\left(c(-\Delta)^{1/2}-\frac{\lambda g^2}{\omega-\Omega}\rho -\omega\right)^{-1}\rho^{1/2}\right\}(\bx,\by)\\ &= \rho(\bx)^{1/2}\rho(\by)^{1/2}\int_0^\infty \d t e^{\omega t} \int d\mu_{\bx,\by,t}(\omega) e^{-\frac{\lambda g^2}{\vert\omega-\Omega\vert}\int_0^t \rho(u(s))\d s }.
\end{align}
Next we make use of the identity
\begin{align}
    (c(-\Delta)^{1/2}-\omega)^{-1} &= \left(c(-\Delta)^{1/2}-\frac{\lambda g^2}{\omega-\Omega}\rho -\omega\right)^{-1}\\ &+ \frac{\lambda g^2}{\vert \omega-\Omega\vert}\left(c(-\Delta)^{1/2}-\frac{\lambda g^2}{\omega-\Omega}\rho -\omega\right)^{-1} \rho \left(c(-\Delta)^{1/2} -\omega\right)^{-1} .
\end{align}
Multiplying this expression on the left and right by $\rho^{1/2}$ we obtain 
\begin{align}
    \frac{\vert\omega-\Omega\vert}{g^2}K_{\omega}[\rho] = A + \lambda AK_{\omega}[\rho],
\end{align}
or equivalently
\begin{align}
    A = \frac{\vert\omega-\Omega\vert}{g^2}K_{\omega}[\rho](1+\lambda K_{\omega}[\rho])^{-1}.
\end{align}
Defining $h(x)=e^{-\lambda x}$ and $F(x) = x(1+\lambda x)^{-1}$, we make use of the result
\begin{align}
    F(x) = x\int_0^{\infty} e^{-y}h(xy) \d y,
\end{align}
to obtain
\begin{align} \label{eq:kernel}
    F(K_{\omega}[\rho])(\bx,\by) = \frac{g^2}{\vert\omega-\Omega\vert}\rho(\bx)^{1/2}\rho(\by)^{1/2}\int_0^\infty \d t e^{\omega t} \int d\mu_{\bx,\by,t}(\omega) h\left(\frac{g^2}{\vert\omega-\Omega\vert}\int_0^t \rho(u(s))\d s\right).
\end{align}
Taking the trace of (\ref{eq:kernel}) we obtain\
\begin{align}
    \tr(F(K_{\omega}[\rho])) &= \frac{g^2}{\vert\omega-\Omega\vert}\int \d \bx \rho(\bx)\int_0^\infty \d t e^{\omega t} \int d\mu_{\bx,\bx,t}(\omega) h\left(\frac{g^2}{\vert\omega-\Omega\vert}\int_0^t \rho(u(s))\d s\right)\\
     &= \int \d \bx \int_0^\infty \d t e^{-1} e^{\omega t} \int d\mu_{\bx,\bx,t}(\omega) f\left(\frac{g^2}{\vert\omega-\Omega\vert}\int_0^t \rho(u(s))\d s\right), \label{eq:trace}
\end{align}
where $f(x) = xh(x)$. Note that $F$ is given in terms of $f$ as
\begin{align}
    F(x) = \int_0^{\infty} e^{-y} y^{-1}f(xy) \d y.
\end{align}
This expression is linear in $f$ and so it can be extended to a wider range of functions. Note that (\ref{eq:trace}) still holds as long as $f$ is a nonnegative, lower semicontinuous function on $[0,\infty)$ such that
\begin{enumerate}
    \item $f(0) = 0$, 
    \item For some $p<\infty$, we have $x^p f(x)\to 0$ as $x\to\infty$.
\end{enumerate}
If in addition $f$ is convex, then by Jensen's inequality we find that
\begin{align}
    \tr(F(K_{\omega}[\rho])) &\leq \int \d \bx \int_0^\infty \d t \ t^{-1} e^{\omega t} \int d\mu_{\bx,\bx,t}(\omega) \frac{1}{t}\int_0^tf\left(\frac{g^2}{\vert\omega-\Omega\vert}t\rho(u(s))\right)\d s\\
    &=\int \d \bx \int_0^\infty \d t \ t^{-1} e^{\omega t} G(\bx,\bx,t)f\left(\frac{g^2}{\vert\omega-\Omega\vert}t\rho(\bx)\right).
\end{align}
Here $G(\bx,\by,t) = e^{-t(-\Delta)^{1/2}}(\bx,\by)$ is the fundamental solution to the fractional heat equation, which is given by
\begin{align}
    e^{-t(-\Delta)^{1/2}}(\bx,\by) = \Gamma\left(\frac{d+1}{2}\right)\frac{1}{\pi^{(d+1)/2}}\frac{t}{(t^2+\vert \bx-\by\vert^2)^{(d+1)/2}},
\end{align}
and satisfies the estimate
\begin{align}
    e^{-t(-\Delta)^{1/2}}(\bx,\by) \leq \frac{C_d}{t^d} ,
\end{align}
for constant $C_d$.
We now define $f(x)$ as
\begin{align}
    f(x) = \begin{cases} 0 & x \leq a , \\
    b(x-a) & x \geq a .
    \end{cases}
\end{align}
We then have
\begin{align}
    N(\rho) \leq \tr(F(K_0)) &\leq \frac{C_d}{c^d}\int\d \bx \int_{t_*}^{\infty}\d t t^{-1}t^{-d}b\left(\frac{g^2}{\Omega}t\rho(\bx)-a\right)\\
    &=\frac{g^{2d}C_d}{\Omega^d c^d}\int \d \bx \rho(\bx)^d,
\end{align}
where $t^*$ is given by 
\begin{align}
    t_{*} = \frac{\Omega}{g^2\rho(\bx)}a.
\end{align}
We can expand the results of Theorem (\ref{thm:NBS}) to the case $\rho\in L^{d}$ by considering a sequence of densities $\{\rho_m\}\subset C_0^{\infty}$ converging to $\rho\in L^{d}$. 
\end{proof}

\section{Computation and properties of the Green's function} \label{app:G}
The Greens's function $G^k(\bx)$  satisfies
\begin{equation}\left((-\Delta)^{1/2} - k \right)G^k(\bx) =  \delta(\bx),\end{equation}
and can be represented as the Fourier integral
\begin{equation}\label{eq:Gint}
	G^k(\bx) = \frac{1}{(2\pi)^d}\int_{\R^d} \frac{e^{\iu \bk\cdot \bx}}{|\bk|-k}\d \bk.
\end{equation}
\subsection{Decay for $k<0$.}
We have the following properties and decay estimate for $G^{k}(\bx)$ when $k<0$.
\begin{lemma}\label{lemma:Greensdecay}
    If $\bx\in\R^d$ and $k < 0$ then $G^k(\bx)$ is real, strictly positive, and
    \begin{align}
        G^k(\bx) \leq \frac{c_d}{k^2\vert\bx\vert^{d+1}}\, .
    \end{align}
    Moreover, $G^k\in L^1(\R^d)$.
\end{lemma}
\begin{proof}
These facts follow immediately from the representation
\begin{align}
    G^k(\bx) &=c_d\int_0^{\infty} e^{kt}\frac{t}{(t^2+\vert\bx\vert^2)^{(d+1)/2}}dt\, , \label{eq:realspacerep}
 \end{align}
with $c_n$ given explicitly by
\begin{align}
    c_d = \frac{\Gamma(\frac{d+1}{2})}{\pi^{\frac{d+1}{2}}}\, .
\end{align}
\end{proof}
\subsection{Computation of the Green's function}
We now evaluate \eqref{eq:Gint} for $d\in \{1,2,3\}$. We begin by observing that 
\begin{equation} \frac{1}{|\bk|-k} - \frac{2k}{|\bk|^2 - k^2} = \frac{1}{|\bk| + k}.\end{equation}
If we let $G^k_\mathrm{helm}$ denote the Helmholtz Green's function, we therefore have the identity
\begin{equation} \label{eq:Gdecomp}
	G^k(\bx) = G^{-k}(\bx) + 2kG^k_\mathrm{helm}(\bx).
\end{equation}
Since there are well-known formulas for $G^k_\mathrm{helm}$, we only need to compute $G^{-k}$.

\subsubsection{One dimension} \label{sec:G1D}
We denote the variable of integration by $\kappa$ in this case. We assume that $\Re(\kappa) > 0$. Then we have
\begin{equation}\frac{1}{2\pi}\int_{\R} \frac{e^{\iu \kappa x}}{|\kappa|+k}\d \kappa = \frac{1}{2\pi}\int_{0}^\infty \frac{e^{\iu \kappa x}}{\kappa+k}\d \kappa + \frac{1}{2\pi}\int_{0}^\infty \frac{e^{-\iu \kappa x}}{\kappa+k}\d \kappa.\end{equation}
Since $G^k(x) = G^k(-x)$, we assume that $x>0$. We can apply Cauchy's integral theorem and Jordan's lemma to deform the contour of integration to the positive and negative imaginary axis, respectively. We then have
\begin{align*}
	G^{-k}(x) &= \frac{1}{2\pi}\int_{0}^\infty \frac{e^{- \kappa}}{\kappa -\iu k x}\d \kappa + \frac{1}{2\pi}\int_{0}^\infty \frac{e^{-\kappa}}{\kappa+\iu k x}\d \kappa \\
	&=\frac{e^{\iu k x }}{2\pi}E_1(\iu kx) + \frac{e^{-\iu k  x }}{2\pi}E_1(-\iu kx).
\end{align*}
Here, $E_1$ is the principal value of the exponential integral, defined as
$$E_1(z) = \int_{z}^\infty\frac{e^{-t}}{t}\d t, \quad z \in \C \setminus \R^-.$$
We have that $E_1(z) + \log(z)$ is holomorphic at $z=0$ and admits the following power series around $z=0$ \cite{abramowitz1964handbook}
\begin{equation}\label{eq:E1}
E_1(z) = -\log z - \gamma - \sum_{n=1}^\infty\frac{(-1)^nz^n}{n!n},
\end{equation}
where  $\gamma$ is the Euler constant and $\log$ denotes the principal branch of the logarithm. Here, $E_1(z)$ has a branch cut along the negative real axis. For  $|z|\to \infty$, we have the asymptotic expansion \cite{abramowitz1964handbook}
\begin{equation}\label{eq:E1asymp}
	E_1(z) = \frac{e^{-z}}{z}\left(1+ \frac{1}{z} + O(|z|^{-2})\right),
\end{equation}
valid uniformly for $z$ away from the negative real axis: $|\arg(z)| \leq \pi - \delta$  for some $\delta >0$.

In one dimension, the Helmholtz Green's function is given by
\begin{equation}G_{\mathrm{helm}}^k(x) = \pm\frac{\iu}{2k}e^{\pm\iu k |x|},\end{equation}
where the sign coincides with the sign of $\Im(k)$. All together, we have from \eqref{eq:Gdecomp} that

\begin{equation}\label{eq:G1}
	G^k(x) = \begin{cases} \ds 
	\frac{1}{2\pi}\left(e^{\iu k |x| }E_1(\iu k|x|) + e^{-\iu k |x| }E_1(-\iu k|x|) \right) + \iu e^{\iu k|x|},\ &\Re(k) > 0, \Im(k)>0, \\[0.5em]
	\ds 
	\frac{1}{2\pi}\left(e^{\iu k |x| }E_1(\iu k|x|) + e^{-\iu k |x| }E_1(-\iu k|x|) \right) - \iu e^{-\iu k|x|},\ &\Re(k) > 0, \Im(k)<0, \\[0.5em]
	\ds \frac{1}{2\pi}\left(e^{\iu k |x| }E_1(\iu k|x|) + e^{-\iu k |x| }E_1(-\iu k|x|) \right),\quad &\Re(k) < 0, \\[0.5em]
	\ds-\frac{1}{\pi} \left(\log(|x|) + \gamma \right) & k = 0,
\end{cases}
\end{equation}
where the case $k=0$ follows from the Fourier transform of $1/|\kappa|$. For purely imaginary $k$, $G^k(x)$ is continuous for $k$ around the positive or negative imaginary axis and can be evaluated as the limit from either side.


\subsubsection{Two dimensions} \label{sec:G2D}
We begin by observing that 
\begin{equation}G^0(\bx) = \frac{1}{(2\pi)^2}\int_{\R^2} \frac{e^{\iu \bk\cdot \bx}}{|\bk|}\d \bk = \frac{1}{2\pi|\bx|}.\end{equation}
Since $$\frac{1}{|\bk|+k} = \frac{-k}{|\bk|(|\bk|+k)} + \frac{1}{|\bk|} ,$$ we have for $\Re(k) >0$ that
\begin{equation}G^{-k}(\bx) = \frac{1}{2\pi|\bx|}-\frac{k}{(2\pi)^2}\int_{\R^2}\frac{e^{\iu \bk\cdot \bx}}{|\bk|(|\bk|+k)}\d \bk.\end{equation}
Changing to polar coordinates we find that
\begin{equation}G^{-k}(\bx) = \frac{1}{2\pi|\bx|}-\frac{k}{(2\pi)^2}\int_{\R^2}\frac{1}{r+k}\int_{0}^{2\pi}e^{\iu r |\bx| \cos\theta} \d \theta \d r .\end{equation}
We now make use of the identities \cite{watson1995treatise}
\begin{equation}J_0(z) = \frac{1}{2\pi}\int_0^{2\pi}e^{\iu z \cos\theta}\d \theta, \qquad \mathbf{K}_0(ak) = \frac{2}{\pi}\int_0^\infty\frac{J_0(ar)}{r+k}\d r,\end{equation}
where $J_0$ is the Bessel function of the first kind of order zero and  $\mathbf{K}_0$ is the Struve function of the second kind of order zero.
We find that 
\begin{equation}\label{eq:G2}
	G^k(\bx) = \begin{cases} \ds 
	\frac{1}{2\pi |\bx|} - \frac{ k}{4}\mathbf{K}_0(k|\bx|) + \frac{\iu k}{2}H_0^{(1)}(k|\bx|),\quad &\Re(k) > 0,\Im(k)>0, \\[0.7em]\ds 
	\frac{1}{2\pi |\bx|} - \frac{ k}{4}\mathbf{K}_0(k|\bx|) - \frac{\iu k}{2}H_0^{(2)}(k|\bx|),\quad &\Re(k) > 0,\Im(k)<0, \\[0.7em]
	\ds\frac{1}{2\pi |\bx|} + \frac{ k}{4}\mathbf{K}_0(-k|\bx|) ,\quad &\Re(k) < 0, \\[0.7em]
	\ds	\frac{1}{2\pi|\bx|}, & k = 0 ,
\end{cases}\end{equation}
where $H_0^{(1)}$ and $H_0^{(2)}$ are the Hankel functions of the first kind and second kind of order zero. 

\subsubsection{Three dimensions}
As in \Cref{sec:G2D}, we have
\begin{equation}G^{-k}(\bx) = \frac{1}{2\pi^2|\bx|^2}-\frac{k}{(2\pi)^3}\int_{\R^3}\frac{e^{\iu \bk\cdot \bx}}{|\bk|(|\bk|+k)}\d \bk.\end{equation}
Changing to spherical coordinates, we find that
\begin{equation}G^{-k}(\bx) = \frac{1}{2\pi^2|\bx|^2}+\frac{\iu k}{4\pi^2}\int_0^\infty\frac{e^{\iu r|\bx|} - e^{-\iu r|\bx|}}{r+k}\d r.\end{equation}
Following the same steps as in \Cref{sec:G1D}, we obtain
\begin{equation}G^{-k}(\bx) = \frac{1}{2\pi^2|\bx|^2}-\frac{\iu k}{4\pi^2}\left(e^{\iu k |\bx| }E_1(\iu k|\bx|) - e^{-\iu k |\bx| }E_1(-\iu k|\bx|)\right).\end{equation}
Since 
\begin{equation}G_\mathrm{helm}^k(\bx) = \frac{e^{\pm\iu k |\bx|}}{4\pi|\bx|},\end{equation}
where the sign is given by the sign of $\Im(k)$, we have
\begin{equation}\label{eq:G3}
	G^k(\bx) = \begin{cases}  
\frac{1}{2\pi^2|\bx|^2}-\frac{\iu k}{4\pi^2|\bx|}\left(e^{\iu k |\bx| }E_1(\iu k|\bx|) - e^{-\iu k |\bx| }E_1(-\iu k|\bx|)\right) +  \frac{ke^{\iu k |\bx|}}{2\pi|\bx|},\ &\Re(k) > 0,\Im(k)>0, \\[0.7em]
 
\frac{1}{2\pi^2|\bx|^2}-\frac{\iu k}{4\pi^2|\bx|}\left(e^{\iu k |\bx| }E_1(\iu k|\bx|) - e^{-\iu k |\bx| }E_1(-\iu k|\bx|)\right) +  \frac{ke^{-\iu k |\bx|}}{2\pi|\bx|},\ &\Re(k) > 0,\Im(k) <0, \\[0.7em]
\ds	 \frac{1}{2\pi^2|\bx|^2}-\frac{\iu k}{4\pi^2|\bx|}\left(e^{\iu k |\bx| }E_1(\iu k|\bx|) - e^{-\iu k |\bx| }E_1(-\iu k|\bx|)\right),\quad &\Re(k) < 0, \\[0.7em]
\ds	\frac{1}{2\pi^2|\bx|^2} & k = 0.
\end{cases}\end{equation}

\subsection{Outgoing Green's function}\label{sec:rad}
For  real $k>0$, we can define two distinct choices  $G^k_\pm$ of the Green's function by taking the limit $\Im(k) \to 0^\pm$ of the corresponding formula \eqref{eq:G1}, \eqref{eq:G2} and \eqref{eq:G3}. When $k>0$, it is straight-forward to check that $G^k_+ = \left(G^k_-\right)^*$.
The two Green's functions $G^k_+$ and $G^k_-$ satisfy the outgoing and incoming Sommerfeld radiation conditions
\begin{equation}\lim_{|\bx|\to \infty} |\bx|^{\frac{d-1}{2}}\left(\frac{\p G_\pm^k}{\p |\bx|} \mp \iu k G^k_\pm \right) = 0.\end{equation}
We can analytically extend these Green's functions to get outgoing and incoming Green's functions,  $G^k_+$ and $G^k_-$, respectively, which are defined for $k$ in the right half-plane $\Re(k) > 0$. In order to select outgoing solutions, we use the outgoing Green's function $G_+^k$ for the integral representation \eqref{eq:psi}. Throughout \Cref{sec:incl}, we omit the subscript and denote the outgoing Green's function as $G^k $: 
\begin{equation}\label{eq:Gout}
	G^k(\bx) = \begin{cases} \ds 
		\frac{1}{2\pi}\left(e^{\iu k |x| }E_1(\iu k|x|) + e^{-\iu k |x| }E_1(-\iu k|x|) \right) + \iu e^{\iu k|x|},  & d= 1,\\[0.5em] \ds
		
		\frac{1}{2\pi |\bx|} - \frac{ k}{4}\mathbf{K}_0(k|\bx|) + \frac{\iu k}{2}H_0^{(1)}(k|\bx|),  & d= 2,\\[0.7em]  
		
		\frac{1}{2\pi^2|\bx|^2}-\frac{\iu k}{4\pi^2|\bx|}\left(e^{\iu k |\bx| }E_1(\iu k|\bx|) - e^{-\iu k |\bx| }E_1(-\iu k|\bx|)\right) +  \frac{ke^{\iu k |\bx|}}{2\pi|\bx|},  & d= 3.
	\end{cases}
\end{equation}

\subsection{Singularity of the Green's function} \label{sec:sing}
We work with the  outgoing Green's function as defined in the previous section. Here we report the behavior of $G$ in the case $\bx \to 0$. Throughout, we let $B\subset \R^d$ be a bounded domain and $\epsilon \ll 1$.

\subsubsection{One dimension}	
	In one spatial dimension, we have for small $\epsilon$ and $x\in B$,
	\begin{equation}\label{eq:Gexp1D}
		G^k(\epsilon x) = \sum_{n=0}^\infty \epsilon^{n}A_n^k(x) +  \sum_{n=0}^\infty \epsilon^{n}\log(\epsilon)\Glog{n}{k}(x),
	\end{equation}
for functions $A_n^k$ and $\Glog{n}{k}$ independent of $\epsilon$, which can be explicitly computed. The first few terms are given by
	\begin{align} \label{eq:Gn1D}
		A_0^k(x) = -\frac{1}{\pi} \left(\log(k|x|)+ \gamma\right) + \iu, \quad \Glog{0}{k}(x) =-\frac{1}{\pi},
	\end{align}
\subsubsection{Two dimensions}
We have that $\mathbf{K}_0(z) = \mathbf{H}_0(z) - Y_0(z)$, where 
\begin{equation}\mathbf{H}_0(z) = \frac{z}{2}\sum_{n=0}^\infty \frac{(-1)^n\left(\frac{z}{2}\right)^{2n}}{\Gamma\left(n+\frac{3}{2}\right)^2}, \qquad Y_0(z) = \frac{2}{\pi}\left(\log\left(\tfrac{z}{2}\right) +\gamma \right)J_0(z) + \sum_{n=1}^\infty a_nz^{2n},\end{equation}
for suitable coefficients $a_n$. Therefore for $x \in B$,
\begin{equation}\label{eq:Gexp2D}
	\epsilon G^k(\epsilon \bx) = \sum_{n=0}^\infty \epsilon^{n}A_n^k(\bx) +  \sum_{n=1}^\infty \epsilon^{n}\log(\epsilon)\Glog{n}{k}(\bx),
\end{equation}
for functions $A_n^k$ and $\Glog{n}{k}$ independent of $\epsilon$, where the first terms are
\begin{equation}\label{eq:Gn2D}A_0(\bx) = \frac{1}{2\pi|\bx|}, \ \Glog{1}{k} = -\frac{k}{2\pi}, \ A_1^k(\bx) = -\frac{k}{2\pi} \left(\log(k|\bx|) + \gamma \right) + \frac{\iu k}{2}, \ \Glog{2}{k} = 0, \ A_2^k = -\frac{k}{\pi}.\end{equation}

\subsubsection{Three dimensions}
Using the power series for $Ein(z)$ in \eqref{eq:E1}, we have the following series expansion of $G$ for $\bx \in B$:
\begin{equation}\label{eq:Gexp}
	\epsilon^2G^k(\epsilon\bx) = \sum_{n=0}^\infty \epsilon^{n}A_n^k(\bx) +  \sum_{n=2}^\infty \epsilon^{n}\log(\epsilon)\Glog{n}{k}(\bx),
\end{equation}
for functions $A_n^k$ and $\Glog{n}{k}$ independent of $\epsilon$. The first few terms are
\begin{equation}\label{eq:Gn3D}
	A_0^k(\bx) = \frac{1}{2\pi^2|\bx|^2}, \ A_1^k = \frac{k}{2\pi|\bx|}, \
	A_2^k(\bx) = \frac{k^2}{2\pi^2}\bigl(1-\gamma + \pi\iu - \log\left(k|\bx|\right) \bigr), \ \Glog{2}{k}(\bx) = -\frac{k^2}{2\pi^2}. 
\end{equation}

\subsubsection{Uniform bounds}
In all three cases $d\in \{1,2,3\}$, it is readily verified that for $\bx \in B$, $k$ in a bounded subset of $\C$ and $n \geq d$, we have 
\begin{equation}
	|A_n^k(\bx) | < C, \quad |\Glog{n}{k}(\bx) | < C, \label{eq:bound}
\end{equation}
for some constant $C$ independent on $\bx, n$ and $k$. In particular, $A_n^k(\bx)$ and $\Glog{n}{k}(\bx)$ are continuous at $\bx = 0$. From \eqref{eq:bound}, it follows that \eqref{eq:Gexp1D}, \eqref{eq:Gexp2D} and \eqref{eq:Gexp} converge uniformly.

\subsection{Decay of the Green's function} \label{sec:decay}
In this section, we present the behavior of $G^k(\bx)$ as $|\bx|\to \infty$ in the case $\Re(k)>0$. Independent of dimension, we have $G^k = G^{-k} + 2kG_\mathrm{helm}^k$, where $G^{-k} \in L^1(\R^d)$. As shown in the following subsections, we have the general behavior for  $d\in \{1,2,3\}$ 
\begin{equation} \label{eq:Gdecay}
	G^k(\bx) = O\left(|\bx|^{\frac{1-d}{2}}e^{\iu k|\bx|}\right)  + O\left(|\bx|^{-(d+1)}\right),
\end{equation}
valid uniformly for $|\arg(k)|\leq \frac{\pi}{2} - \delta$ and $|k| <K$ for some positive constants $\delta$ and $K$. Specifically, for $\Im(k)>0$, the exponential term decays and $G^k(\bx) = O\left(|\bx|^{-(d+1)}\right).$


\subsubsection{One dimension}
When $d=1$, we can combine \eqref{eq:Gout} with \eqref{eq:E1asymp} to find, as $|x|\to \infty$,
\begin{equation}G^k(x) =  \iu e^{\iu k |x|} + O(x^{-2}).\end{equation}

\subsubsection{Two dimensions}
For large $z$ we have the well-known asymptotics of the Struve function \cite{abramowitz1964handbook}
\begin{equation}\mathbf{K}_0(z) = \frac{2}{\pi z} + O(z^{-3}),\end{equation}
for $z$ away from the real axis: $|\arg(z)| \leq \pi - \delta$  for some $\delta >0$. Combined with \eqref{eq:Gout} for $d=2$, we have  as $|\bx| \to \infty$
\begin{equation}G^{k}(\bx) =  \frac{\iu k}{2}H_0^{(1)}(k|\bx|) + O(|\bx|^{-3}).\end{equation}
From the well-known asymptotics \cite{abramowitz1964handbook}
\begin{equation}
	H_0^{(1)}(k|\bx|) = O\left(|\bx|^{-\frac{1}{2}}e^{\iu k |\bx|}\right),
\end{equation}
\eqref{eq:Gdecay} readily follows for $d=2$.

\subsubsection{Three dimensions}
When $d=3$, we again use the asymptotics of $E_1$ given in \eqref{eq:E1asymp} which gives
\begin{equation}-\frac{\iu k}{4\pi^2|\bx|}\left(e^{\iu k |\bx| }E_1(\iu k|\bx|) - e^{-\iu k |\bx| }E_1(-\iu k|\bx|)\right) = -\frac{1}{2\pi^2|\bx|^2} + O\left(|\bx|^{-4}\right),\end{equation}
as $|x|\to \infty$. Together with \eqref{eq:Gout}, we therefore have
\begin{equation}\label{eq:decay}
	G^k(\bx) = \frac{ke^{\iu k|\bx|}}{2\pi|\bx|}+O\left(|\bx|^{-4}\right).
\end{equation}


\section{Gohberg-Sigal theory and proof of \Cref{prop:pert}}\label{sec:GS}
In this section, we present the proof \Cref{prop:pert}. The arguments are based on Gohberg-Sigal perturbation theory for holomorphic operator-valued functions. For a detailed presentation of this theory we refer to \cite{ammari2018mathematical,Gohberg1971}. We will adopt the definitions as introduced in \cite[Chapter 1]{ammari2018mathematical}. 
\begin{defn}
	Let $A(z)$ be a holomorphic operator-valued function of $z\in V\subset\C$. The point $z_0\in V$ is called a \emph{characteristic value} of $A$ if there is a holomorphic vector-valued function $\phi(z)$ such that $\phi(z_0) \neq 0$ and $A(z_0)\phi(z_0) = 0$.
\end{defn}
This terminology agrees with the previous literature; observe, however, that characteristic values are called \emph{nonlinear eigenvalues}, or just \emph{eigenvalues}, throughout the main part of this paper.
We say that a characteristic value is \emph{simple} if $\dim \ker A(z_0) = 1$ and 
\begin{equation}A(z)\phi(z) = (z-z_0)\psi(z), \quad \psi(z_0) \neq 0,\end{equation}
for some vector-valued holomorphic function $\psi(z)$. 

\begin{prop}	
	Let $A_\epsilon(z)$ be a family of holomorphic operator-valued functions of $z\in V\subset\C$, continuous for $z\in \overline{V}$ and invertible for all $z\in \partial V$, and which depends continuously on $\epsilon$ in the operator norm. Assume that $A_0(z)$ has a single (counting with multiplicity) characteristic value $z_0\in V$ and that $A_\epsilon(z_0)$ is Fredholm of index zero for all $\epsilon$ small enough. Then, for small enough $\epsilon$, there exist a single eigenvalue (counting with multiplicity) $z =  z(\epsilon) \in \C$ of $A(z)$ satisfying $\lim_{\epsilon \to 0^+}z(\epsilon) = z_0$.
\end{prop}
\begin{proof}
	Let \begin{equation}E(z) = \bigl(A_\epsilon(z) - A_0(z)\bigr).\end{equation}
	Given some small $\delta$, we let $V_\delta$ be the disk of radius $\delta$ and center $z_0$. Then, for small enough $\epsilon$ we have that, in the operator norm,
	\begin{equation}\left\|\bigl(\A_0(z)\bigr)^{-1} E(z)\right\| < 1,\end{equation}
	for all $z \in \p V_\delta$. From Rouché's theorem, \cite[Lemma 1.11 and Theorem 1.15]{ammari2018mathematical}, we find that $A(z)$ has a single eigenvalue $z(\epsilon)$ inside $V_\delta$. Moreover, for any $\delta$ we have $|z(\epsilon) - z_0| < \delta$ for all $\epsilon$ small enough, so 
	\begin{equation}
		\lim_{\epsilon\to 0} z(\epsilon) = z_0. \qedhere
	\end{equation}
\end{proof}
Using the argument principle, we also have an explicit expansion of the characteristic values (see \cite[Theorem 3.7]{ammari2018mathematical}).
\begin{theorem}\label{thm:gap}
	Let $A_\epsilon(z)$ be a family of holomorphic operator-valued functions of $z\in V\subset\C$, continuous for $z\in \overline{V}$ and invertible for all $z\in \partial V$, and which depends continuously on $\epsilon$ in the operator norm. Assume that $A_0(z)$ has a single (counting with multiplicity) characteristic value $z_0\in V$ and that $A_\epsilon(z_0)$ is Fredholm of index zero for all $\epsilon$ small enough. Then, for small enough $\epsilon$, there is a single characteristic value of $A_\epsilon(z)$ in $V$ satisfying
	\begin{equation}
	z_\epsilon  - z_0 = \frac{1}{2\pi \iu}\sum_{p=1}^\infty\frac{1}{p}
	\tr\int_{\partial V}\Bigl[\A_0(z)^{-1}\bigl(\A_0(z) - \A_\epsilon(z) \bigr)\Bigr]^p\d z.
	\end{equation}
\end{theorem}
We now turn to the operator $\G_\epsilon^\omega$ as introduced in \Cref{sec:incl}.
\begin{lemma}\label{lem:holo}
	Let $d\in \{1,2,3\}$, $U = \C \setminus \iu \R$. The operator
	\begin{equation}\G_\epsilon^\omega: L^2(D) \rightarrow L^2(D),\end{equation}
defined in \eqref{eq:A}	is a holomorphic operator-valued function of $\omega\in V$ and depends continuously on $\epsilon$. Moreover, if $\omega\neq \Omega$, then $\G_\epsilon^\omega$ is Fredholm with index 0.	
\end{lemma}
\begin{proof}
	For $\bx\neq 0$, the integral kernel $G^{k}$ is a holomorphic function of $k \in V_0$ so that, for $k_0 \in V_0$ and $k$ in a neighbourhood of $k_0$,
	\begin{equation}G^{k}(\bx) = \sum_{n=0}^\infty (k-k_0)^ng_n(\bx).\end{equation}
	We need to verify that the integral operators associated to $g_n$ exist and are uniformly bounded. Since $G^{k}(\bx) = G^{-k}(\bx) + 2kG_\mathrm{helm}^k(\bx)$
	\begin{equation}G^{-k}(\bx) = \sum_{n=0}^\infty (k-k_0)^na_n(\bx), \quad 2kG^{k}_\mathrm{helm}(\bx) = \sum_{n=0}^\infty (k-k_0)^nb_n(\bx).\end{equation}
	Observe that 
	\begin{equation}a_n(\bx) = \frac{1}{n!}\frac{\p^n}{\p k^n}\int_{\R^d} \frac{e^{\iu \bk\cdot\bx}}{|\bk|+k} \d \bk.\end{equation}
	At $\bx = 0$, $a_n(\bx)$ is weakly singular for $n < d$ and continuous and uniformly bounded for $n\geq d$ and $\bx \in D$. Similarly, $b_n(\bx)$ is weakly singular for $n=0$ and continuous and uniformly bounded for $n\geq 1$ and $\bx \in D$. In summary we have 
	\begin{equation}g_n(\bx) = a_n(\bx) + b_n(\bx),\end{equation}
	and $g_n$ defines a family of operators $A_n:L^2(D) \to L^2(D)$ such that, for $\omega$ in a neighbourhood of $\omega_0$ we have 
	\begin{equation}\G_\epsilon^\omega = A_0-(\omega_0-\Omega)I + (\omega-\omega_0)(A_1-I) + \sum_{n=2}^\infty (\omega-\omega_0)^nA_n, \quad A_n = -\frac{g^2s_0}{\epsilon c}\int_{D} g_n(\bx-\by) \phi(\by) \d \by,\end{equation}
	where the sum converges in $\B\bigl(L^2(D)\bigr)$. Hence $\G_\epsilon^\omega$ is holomorphic in $V_0$. From the expansions of $G^k$ in \Cref{sec:sing} we know that $\G_\epsilon^\omega$ is continuous for $\epsilon$ around $0$. Moreover, if $\omega_0 \neq \Omega$, then
	\begin{equation} \G_\epsilon^{\omega_0} = A_0-(\omega_j-\Omega)I\end{equation}
	is a compact perturbation of the identity, and is therefore a Fredholm operator of index $0$.
\end{proof}

\section*{Acknowledgments}
We thank Jeremy Hoskins for valuable discussions, especially related to the calculation of the Green's functions in Appendix \ref{app:G}. JCS was supported in part by the NSF grant DMS-1912821 and the AFOSR grant FA9550-19-1-0320. MIW was supported in part by 
 NSF grant DMS-1908657 and Simons Foundation Math + X Investigator Award \# 376319.

\nocite{*}
\bibliographystyle{abbrv}
\bibliography{one-photon}{}

\end{document}